\documentclass[journal]{IEEEtran}
\usepackage{amsmath, amssymb, amsthm}
\usepackage{graphicx}
\usepackage{float}
\newfloat{figtab}{htb}{fgtb}
\makeatletter
  \newcommand\figcaption{\def\@captype{figure}\caption}
  \newcommand\tabcaption{\def\@captype{table}\caption}
\makeatother

\newtheorem{Proposition}{Proposition}

\usepackage{multicol}
\usepackage{algorithmic,algorithm}
\allowdisplaybreaks[4]
\def\BibTeX{{\rm B\kern-.05em{\sc i\kern-.025em b}\kern-.08em
    T\kern-.1667em\lower.7ex\hbox{E}\kern-.125emX}}
\usepackage{balance}
\begin{document}
\title{Proactive Content Caching Scheme in Urban Vehicular Networks}
\author{
\IEEEauthorblockN{Biqian Feng, Chenyuan Feng, Daquan Feng, Yongpeng Wu, and Xiang-Gen Xia}
\thanks{B. Feng, and Y. Wu are with the Department of Electronic Engineering, Shanghai Jiao Tong University, Minhang 200240, China (e-mail: fengbiqian@sjtu.edu.cn; yongpeng.wu@sjtu.edu.cn).}
\thanks{C. Feng, and D. Feng are with  Guangdong-Hong Kong Joint Laboratory for Big Data Imaging and Communication, Shenzhen University, Shenzhen 518060, China (email: fengchenyuan@szu.edu.cn, fdquan@szu.edu.cn). \textit{(Corresponding author: Daquan Feng.)} }
\thanks{X.-G. Xia is with the Department of Electrical and Computer Engineering, University of Delaware, Newark, DE 19716, USA. (e-mail: xxia@ee.udel.edu).}
}

\maketitle

\begin{abstract}
Stream media content caching is a key enabling technology to promote the value chain of future urban vehicular networks. Nevertheless, the high mobility of vehicles, intermittency of information transmissions, high dynamics of user requests, limited caching capacities and extreme complexity of business scenarios pose an enormous challenge to content caching and distribution in vehicular networks. To tackle this problem, this paper aims to design a novel edge-computing-enabled hierarchical cooperative caching framework. Firstly, we profoundly analyze the spatio-temporal correlation between the historical vehicle trajectory of user requests and construct the system model to predict the vehicle trajectory and content popularity, which lays a foundation for mobility-aware content caching and dispatching. Meanwhile, we probe into privacy protection strategies to realize privacy-preserved prediction model. Furthermore, based on trajectory and popular content prediction results, content caching strategy is studied, and adaptive and dynamic resource management schemes are proposed for hierarchical cooperative caching networks. Finally, simulations are provided to verify the superiority of our proposed scheme and algorithms. It shows that the proposed algorithms effectively improve the performance of the considered system in terms of hit ratio and average delay, and narrow the gap to the optimal caching scheme comparing with the traditional schemes.
\end{abstract}

\section{Introduction}
With the gradual improvement of the degree of autonomous driving, the demand for in-vehicle entertainment services has been increasing. However, the high mobility of vehicles, intermittence of information transmission, high dynamics of popular content, limitations of cache capacity and the complexity of business scenarios bring great challenges to hot content prediction, multi-node collaborative cache distribution and service quality optimization. Effective caching system has attracted numerous scholars' concern in terms of vehicle trajectory prediction, popular content prediction, content placement and content delivery strategies.

\subsection{Related Works}
First of all, vehicle trajectory prediction plays a critical role in caching systems due to the high speed of vehicles and limited communication range of vehicle-to-infrastructure (V2I) links. In \cite{ref01}, Gaussian model and Long Short-Term Memory (LSTM) model are proposed  to predict vehicle trajectory. By extending \cite{ref01}, lots of variants of Markov model are proposed for location prediction, including $N$-order Markov model, hidden Markov model, and variable-order Markov model. Specifically, the $N$-order Markov model \cite{ref01_1} and hidden Markov model \cite{ref01_2} utilize the state transition matrix to predict the vehicles' future locations by computing the transition probability. In \cite{ref02, ref02_1}, variable-order Markov models are designed to solve the prediction problem by the help of Prediction by Partial Match (PPM) and Probabilistic Suffix Tree (PST) algorithms. However, the above-mentioned algorithms fail to intelligently distinguish the importance of the track data in different historical periods. 

As traditional passive caching schemes are becoming more and more unsuitable for the era of information explosion, proactive caching based on popular content prediction is proposed as a promising solution, in which the recommendation systems \cite{SASRec, NCF} are used to model the relationship between users and content and improve the prediction accuracy of user preferences \cite{ref03}. Recently, federated learning based method is used to improve the performance of context-aware popularity prediction scheme \cite{ref04}. However, the above-mentioned works ignore the impact of user mobility. In addition, LSTM-based two-tier cache architecture is proposed to cope with the user mobility, in which the high-speed and low-speed users are served by macro stations and small base stations, respectively, so as to avoid frequent switching of highly dynamic users \cite{ref05}. Although existing researches have made good progress in popular content prediction problem, they ignore the protection of user privacy. 

As for the content caching mechanism in vehicular networks, improving the utilization of storage space of Roadside units (RSUs) has attracted attention of researchers. In \cite{ref05_1}, the authors assume that the vehicle user requests are already known by a cache-aided network, and then propose a novel distributed caching strategy based on Gibbs sampling to optimize the cache hit probability. In \cite{ref06}, the block matrix method is used to extract users' preferences based on their historical interests in videos and select appropriate RSU to cache corresponding content. Besides, deep learning method, such as Q-learning algorithm, is also proposed to effectively improve the quality of service (QoS) within limited resources \cite{ref07, newref1, newref2}. In \cite{ref07_1}, the multi-agent reinforcement learning (MARL) are adopted by all wireless network nodes to collaboratively optimize the distributed caching strategy and maximize the network performance, which are measured by the average cache hit probability.

At last, existing works related to content distribution mechanism in vehicular networks can be divided into three categories, namely: mechanisms based on Vehicle-to-Infrastructure (V2I) communications \cite{ref08}, based on Vehicle-to-Vehicle (V2V) communications \cite{ref09}, based on collaborative communications of V2I and V2V \cite{ref10}. It is worth noting that most of related works assume vehicle trajectory data and user requests are known and lack consideration of the video coding characteristics, such as coding structure and bit rate. 

In brief, lots of state-of-the-art works have been carried out to improve the performance of multimedia content distribution, however, they lack comprehensive consideration of the inherent characteristics of vehicular networks, video coding characteristics, user service demands, and the difference analysis of business scenarios.

\subsection{Motivation and Contributions }
Motivated by these issues, we aim to integrate edge computing into the vehicular networks, and propose a framework of content caching and distribution to improve the quality of service (QoS), protect user privacy and also achieve a high resource utilization efficiency. To this end, we firstly build an integrated service framework for vehicle trajectory prediction and privacy-persevered popular content prediction based on deep analysis of the spatio-temporal correlation between vehicle trajectory and user requests. Furthermore,  we design a mobility-aware and business-adaptive algorithm for collaborative caching scheme based on optimization algorithms. The main contributions of this paper are summarized as follows.
\begin{enumerate}
\item We propose a Hierarchical Cooperative Caching Network (HCCN) architecture which consists of three layers, namely, the central cloud service layer, edge computing layer, and terminal equipment layer. The periodical processing procedure can be distinguished as three main execution phases: trajectory prediction, content popularity prediction, and content caching, which can adapt to the dynamic properties of vehicular ad hoc networks (VANETs) topology, provide real-time content popularity prediction, and reduce communication costs. Furthermore, a pipeline scheduling mechanism is proposed for parallel execution of prediction and transmission, which can reduce the service delay and improve the quality of experience (QoE).

\item We will make the utmost of the spatio-temporal correlation of historical trajectory data and design an LSTM-based model to predict the residence time in each RSU in the future. Specifically, the model extracts daily features from the daily trajectory, and fuses daily features into the historical feature information. Finally, the future trajectory prediction module aims at predicting the future residence time in each RSU by combining intraday trajectory and historical feature information.

\item Since the recent behavior can reveal vehicles' future preferences to a certain extent, we modify the self-attentive sequential recommendation (SASRec) model to predict future content requests. Furthermore, with the growing concern on data privacy and consideration of increasing on-board training data, we propose a Hierarchical Federated Learning (HFL)-based structure to train the SASRec network for each cluster. Hence, the content popularity of each RSU and macro base station (MBS) can be naturally derived based on the requirements of all connecting vehicles. 

\item Based on the aforementioned trajectory prediction and content popularity prediction results, we formulate an optimization problem for dynamic cooperative content caching. However, it is a large-scale 0-1 constrained problem, which is NP-hard in nature. To tackle it efficiently, we propose an adaptive gradient descent algorithm to enhance the performance of content caching, which is verified to perform well by our simulation results.
\end{enumerate}

The rest of this paper is organized as follows. Section \ref{HCCN} introduces our proposed HCCN architecture that is utilized to establish low-latency networks in Section \ref{section: system model}. Section \ref{Trajectory Prediction Scheme} depicts a vehicle trajectory prediction scheme. Section \ref{Recommendation System Scheme} proposes an HFL-based SASRec network to predict content popularity. Section \ref{section: dynamic content caching scheme} integrates trajectory prediction and content popularity prediction into dynamic and cooperative content caching scheme, and proposes an adaptive gradient descent algorithm to solve the large-scale 0-1 constrained problem. Section \ref{Simulation} provides some simulation results to evaluate the performances of our proposed schemes and algorithms. Section \ref{Conclusion} concludes the paper.

The notations used in this paper are as follows. Boldface lowercase and uppercase letters, such as $\mathbf{a}$ and $\mathbf{A}$, are used to represent vectors and matrices, respectively. Superscript $T$ stands for the transpose,  $\mathbb R$ is the set of real numbers, $\nabla L$ denotes the gradient of $L$ and $\left(\nabla L\right)_{\mathbf x}$ represents its $\mathbf x$-th component.

\section{Overall Design of the HCCN Architecture}
\label{HCCN}
In this section, we will introduce the overall design of our proposed HCCN architecture and the periodical processing procedure in detail. 
\subsection{Content Retrieving Scheme}
\begin{figure*}[htbp]
\centering
\includegraphics[scale=0.035]{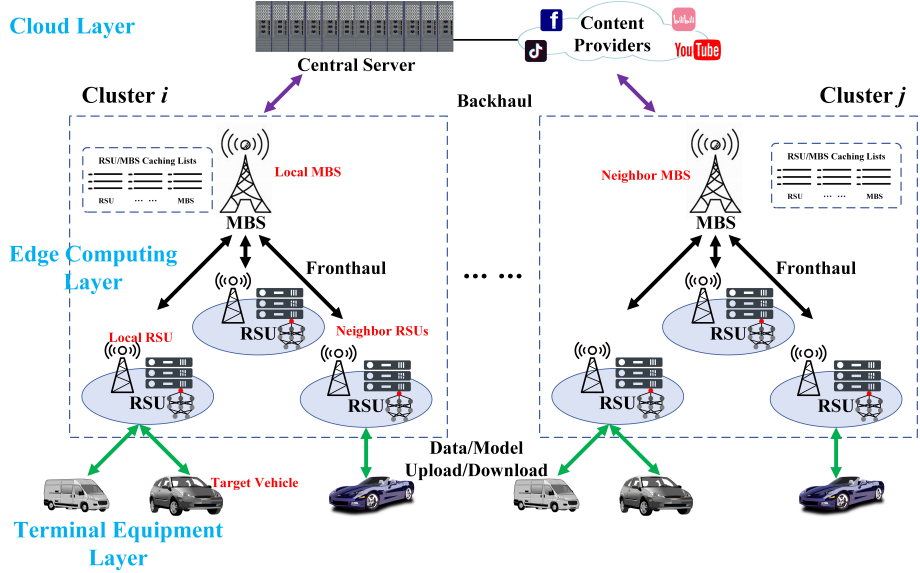}
\caption{An example of the edge computing-enabled content caching system in hierarchical vehicular networks, each edge node holds a caching content list and performs collaborative caching with its connecting instructions and devices in a federated learning manner.}
\label{system_architecture}
\end{figure*}

A novel network-level content caching protocol will be designed at first. As shown in Fig. \ref{system_architecture}, the hierarchical architecture consists of the following three layers:
\begin{itemize}
\item Cloud layer: it contains content providers, such as TikTok and YouTube, and cloud computing server to provide contents and computing services.

\item Edge computing layer: it contains all edge nodes, namely, RSUs, macro base stations (MBSs) and baseband unit (BBU) pools, where each MBS and multiple RSUs within its coverage area can form a cluster. In terms of communication, all MBSs can connect to each other and the central cloud through optical fibers, and communicate with the RSUs in their cluster through wireless links. As for caching, each RSU will send a caching list including the identification and location of caching contents to their connecting MBS,  all MBSs will merge the collected caching lists and exchange with each other, by this means, all MBSs have the same caching content lists containing the identification and location of  contents cached by all RSUs and MBSs, which facilitates mutual retrieval of cached content conveniently.

\item Terminal equipment layer: it contains all the vehicles and intelligent devices that need to be served along the road.
\end{itemize}

The proposed HCCN intends to maximize the network performance by leveraging the vertical cooperation among the MBSs and their connecting RSUs, the horizontal cooperation among the local RSU and its neighbor RSUs, and also among the local MBS and its neighbor MBSs. Specifically, when a vehicle sends a content request to its local RSU, the local RSU look through its own cache list to check whether the requested content is stored or not. If cached, the requested content will be transmitted directly to the vehicle from the local RSU. Otherwise, the local RSU will ask local MBS to check its caching list that contains the identification and location of contents cached by all RSUs and MBSs. The local MBS will search the required content according to the following order: firstly, the local MBS/RSUs in the same cluster, then, other MBSs/RSUs outside the current cluster, lastly, the cloud. If cached, the local MBS will fetch the requested content from the source node, and then forward it to the target RSU. Once received, the local RSU will send it to the target vehicle. The requested content can be provided by caching at either MBSs or RSUs, therefore, it greatly reduces the congestion between the target vehicles and the core network. Otherwise, the local MBS needs to send the request to the Internet and obtain the content from the source (i.e., content provider) in the cloud. In our HCCN framework, massive content requests for the same hot contents not only greatly relieve the burden on the core network, but also reduce the vehicles' service delay and improves the QoE of the vehicles.

\subsection{Pipeline Scheduling Mechanism}
As shown in Fig. \ref{pipeline_structure}, we propose a parallel pipeline scheduling mechanism, where the content service is executed periodically based on three stages: prediction, caching, and transmission. 
Firstly, the prediction phase includes vehicle trajectory and content popularity predictions. Trajectory prediction intends to predict the residue time in each RSU for each vehicle, while content popularity prediction aims at discerning the popular contents that will be required by the vehicles in the near future. 
Secondly, in the caching stage, all RSUs and MBSs implement the mobility- and popularity-aware proactive edge caching scheme to pursue a higher network resource utilization and provide users with better QoE. 
Finally, based on content caching and vehicles' characteristics, the edge computing layer performs an adaptive distribution mechanism for multimedia content service by dynamically configuring time-frequency resources. 
There are mainly two typical situations of real-time service in the transmission stage: i) the vehicle sends a request to the local RSU, then the local RSU has ability to obtain and send the requested content as soon as possible; ii) the local RSU can proactively provide some personalized contents for each vehicle based on its caching contents. Based on the results in the prediction stage and the caching stage, the content retrieval delay in both situations is greatly reduced in the transmission stage.
After making predictions and cache deployment decisions, edge nodes can execute the prediction phase of the next episode in parallel with content transmission phase of the current episode, which can take full advantage of  the spatial-temporal correlation among trajectory data and content popularity. By this means, the proposed mechanism earns a shorter service delay and higher efficiency, compared with the traditional serial scheduling manner. 
\begin{figure*}[htbp]
\centering
\includegraphics[scale=0.05]{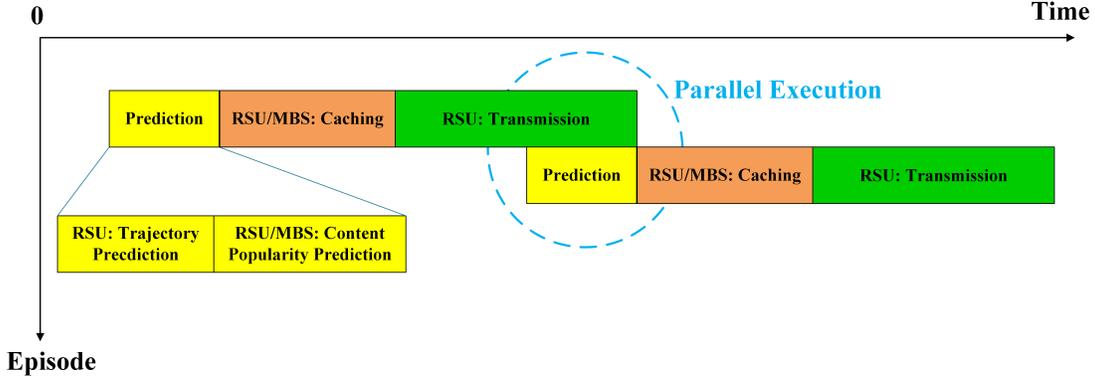}
\caption{An example of the parallel pipeline scheduling mechanism for prediction, caching and transmission, the edge nodes could execute the new prediction phase in parallel with current content transmission phase. }
\label{pipeline_structure}
\end{figure*}
 
\section{System Model}
\label{section: system model}
To effectively leverage the advantage of the HCCN architecture described in Section \ref{HCCN}, we intend to formulate a cooperative caching problem to minimize the content prefetching latency in urban roads in this section.

\subsection{Network Architecture}
As shown in Fig. \ref{system_architecture}, we consider a vehicular edge computing network with three different types of edge caching nodes, including MBSs, RSUs, and vehicle nodes. Let $\mathcal M=\left\{1,2,\cdots,M\right\}$, $\mathcal R=\left\{1,2,\cdots, R\right\}$, and $\mathcal V=\left\{1,2,\cdots,V\right\}$ represent the index sets of MBSs, RSUs, and vehicle nodes, respectively. In the urban road network, RSUs are deployed intensively for the high traffic flow. According to physical locations, MBS $m$ can manage a group of RSUs , $\mathcal R_m\subseteq\mathcal R$, within its coverage area. In this work, we define a cluster as one MBS and its connecting RSUs. Since the transmission cost of cellular communications is much higher than that of vehicle-to-RSU (V2R) communications, vehicles prefer to retrieve contents from nearby RSUs. Let $\mathcal F=\left\{1, 2, \cdots, F\right\}$ denote the index set of files provided by content providers and each content $f\in\mathcal F$ has a size of $s_f$. Since MBSs and RSUs are equipped with limited storage capacities, let $S_m^{\text{MBS}}$ and $S_r^{\text{RSU}}$ denote the caching spaces of MBS $m$ and RSU $r$, respectively, $\mathcal F_m^\text{MBS}\subseteq \mathcal F$ and $\mathcal F_r^\text{RSU} \subseteq \mathcal F $ denote the content sets stored by MBS $m$ and RSU $r$, respectively.

\subsection{Content Caching Policy}
\label{Caching Decision}
To meet the requirements of the transmission stage at time slots $\mathcal T=\left\{1,2,\cdots, T\right\}$ in Fig. \ref{pipeline_structure}, the contents should be collaboratively cached by target RSUs and MBSs in advance. Let $\mathbf x_r=\left(x_{r,1}, x_{r,2}, \cdots, x_{r,F}\right)^T$ denote the caching decision vector of RSU $r$, where $x_{r,f}\in\left\{0, 1\right\}$ is a Boolean variable to indicate caching placement decision, namely, $x_{r,f}=1$ if file $f$ is cached by RSU $r$, otherwise, $x_{r,f}=0$. Since the total size of cached files cannot exceed the entire storage capacity of RSUs, $\mathbf x_r$ must satisfy the following constraint:
\begin{equation}
    \sum_{f\in\mathcal F}x_{r,f} s_f\leq S_r^\textrm{RSU}.
\end{equation}
Similarly, let vector $\mathbf y_m=\left(y_{m,1}, y_{m,2}, \cdots, y_{m,F}\right)^T$ represent the caching decision of MBS $m$, which should satisfy the following constraints:
\begin{equation}
    \sum_{f\in\mathcal F}y_{m,f} s_f\leq S_m^\text{MBS}.
\end{equation}
The content retrieval delay is generally positively correlated with the distance from the source node to the destination node. Define $\gamma^{CM}, \gamma^{MR}$, and $\gamma^{MM}$ as the transmission rate of backhaul links between the cloud and MBS, fronthaul links between the MBS and its connecting RSU, and the links between two MBSs, respectively. Apparently, $\gamma^{MR}, \gamma^{MM} \gg \gamma^{CM}$. The total content retrieval delay is given by:
\begin{equation}
    \label{total content retrieval delay1}
    \begin{aligned}
        \gamma_{r,f}=\gamma_{r,f}^0+\gamma_{r,f}^1+\gamma_{r,f}^2+\gamma_{r,f}^3+\gamma_{r,f}^4+\gamma_{r,f}^5,
    \end{aligned}
\end{equation}
where $\gamma_{r, f}^{i} $ denotes the retrieval delay of content $f$ fetched by RSU $r$ from its own cache if $i = 0$, from its local MBS if $i = 1$, from other RSUs within the same cluster if $i =2$, from the other MBSs if $i=3$, from other RSUs outside its cluster if $i=4$, and from the cloud network if $i=5$. Specifically, they are determined by the size of content, the transmission rate of all links, and the caching decision $x_{r,f}$ and $y_{r,f}$:
\begin{align}
    \label{content retrieval delay}
    \gamma_{r, f}^{0}&=0,\quad \gamma_{r, f}^{1}=\frac{s_{f}}{\gamma^{MR}} \left(1-x_{r, f}\right) y_{m, f}, \notag\\
    \gamma_{r, f}^{2}&=2 \frac{s_{f}}{\gamma^{M R}}\left(1-x_{r, f}\right)\left(1-y_{m, f}\right)\notag\\
    &\qquad\qquad\bigg[1-\prod_{r^{\prime} \in \mathcal{R}_{m}, r^{\prime} \neq r}\left(1-x_{r^{\prime}, f}\right)\bigg], \notag\\
    \gamma_{r, f}^{3}&=\left(\frac{s_{f}}{\gamma^{M M}} + \frac{s_{f}}{\gamma^{M R}}\right)\left(1-y_{m, f}\right) \notag\\
    &\qquad\qquad\prod_{r^{\prime} \in \mathcal R_m}\left(1-x_{r^{\prime}, f}\right)\bigg[1-\prod_{m^{\prime} \neq m}\left(1-y_{m^{\prime}, f}\right)\bigg], \\
    \gamma_{r, f}^{4}&=\left(\frac{s_{f}}{\gamma^{M M}} + 2\frac{s_{f}}{\gamma^{M R}}\right) \prod_{r^{\prime} \in \mathcal{R}_{m}}\left(1-x_{r^{\prime}, f}\right)  \notag\\
    &\qquad\qquad\prod_{m^{\prime} \in \mathcal{M}}\left(1-y_{m^{\prime}, f}\right)\bigg[1-\prod_{r^{\prime} \notin \mathcal{R}_{m}}\left(1-x_{r^{\prime}, f}\right)\bigg], \notag\\
    \gamma_{r, f}^{5}&=\left(\frac{s_{f}}{\gamma^{C M}} + \frac{s_{f}}{\gamma^{M R}}\right) \prod_{r^{\prime} \in \mathcal{R}}\left(1-x_{r^{\prime}, f}\right) \prod_{m^{\prime} \in \mathcal{M}}\left(1-y_{m^{\prime}, f}\right),\notag
\end{align}
If the content is cached in the local RSU, then the local RSU forwards it to the vehicle directly, thus the delay is $0$; If the content $f$ is not cached in the local RSU but cached in the local MBS, i.e., $1-x_{r, f}=1$ and $y_{m, f}=1$, then $\gamma_{r, f}^{1}>0$ from content retrieving link MBS-RSU and $\gamma_{r, f}^{i}=0, i\neq 1$; Similarly, the content $f$ fetched in other nodes can be represented by $\gamma_{r, f}^{i}, i=2,3,4$. Substituting \eqref{content retrieval delay} into \eqref{total content retrieval delay1}, the total content retrieval delay can be rewritten as
\begin{equation}
\label{total content retrieval delay}
\begin{aligned}
&\gamma_{r,f}=\gamma^{MR} s_{f}\bigg[\left(1-x_{r, f}\right)+\left(1-x_{r, f}\right)\left(1-y_{m, f}\right)\\
&\quad+\prod_{r^{\prime} \in \mathcal R_m}\left(1-x_{r^{\prime}, f}\right)\prod_{m^{\prime} \in \mathcal{M}}\left(1-y_{m^{\prime}, f}\right)\bigg]\\
&\quad+\left(\gamma^{M M} - \gamma^{M R}\right)s_f\prod_{r^{\prime} \in \mathcal R_m}\left(1-x_{r^{\prime}, f}\right)\left(1-y_{m, f}\right)\\
&\quad+\left(\gamma^{CM} - \gamma^{M M} - \gamma^{M R}\right)s_f\prod_{r^{\prime} \in \mathcal R}\left(1-x_{r^{\prime}, f}\right)\\
&\quad\prod_{m^{\prime} \in \mathcal{M}}\left(1-y_{m^{\prime}, f}\right).
\end{aligned}
\end{equation}
Since $\gamma^{MR}$, $\gamma^{CM}$, $\gamma^{MM}$, and $s_f$ are fixed in a system, the above function $\gamma_{r,f}$ is a function of variables $\mathbf x$ and $\mathbf y$. Let binary variables $\theta_{v,r,t}^1 \in \left\{0, 1\right\}$ and $\theta_{v,f,t}^2 \in \left\{0, 1\right\}$ indicate whether user $v$ enters the coverage of RSU $r$ at time slot $t$, and whether user $v$ requests for the video $f$ at time slot $t$, respectively. Therefore, the transmission cost of user $v$ retrieving content $f$ by RSU $r$ at time slot $t$ is $L_{u,r,f,t}\triangleq \theta_{v,r,t}^1 \theta_{v,f,t}^2 \gamma_{r,f}$. Assuming that user interests in a certain content are independent to their locations, then the expected cost of caching content $f$ by RSU $r$ is given by 
\begin{equation}
\begin{aligned}
	&\mathbb E\left(L_{v,r,f,t}\right) = \mathbb E\left[\theta_{v,r,t}^1\right] \mathbb E\left[\theta_{v,f,t}^2\right]\gamma_{r,f},
\end{aligned}
\end{equation}
where $\mathbb E\left[\theta_{v,r,t}^1\right]$ and $\mathbb E\left[\theta_{v,f,t}^2\right]$ can be considered as the probabilities of vehicle $v$ staying in the coverage of RSU $r$ and retrieving the content $f$ at time slot $t$, respectively. We consider the users' interests will not change in a short time, i.e., $\mathbb E\left[\theta_{v,f,t}^2\right]$ keeps the same in a scheduling duration. For simplicity, we drop the time slot subscript $t$ and the expected cost is restated as follows:
\begin{equation}
\begin{aligned}
	&\mathbb E\left(L_{v,r,f,t}\right) = \mathbb E\left[\theta_{v,r,t}^1\right] \mathbb E\left[\theta_{v,f}^2\right]\gamma_{r,f}.
\end{aligned}
\end{equation}
Furthermore, the total expected caching cost of the system is shown as follows:
\begin{equation}
\label{Wxy}
\begin{aligned}
	&W\left(\mathbf x,\mathbf y\right)\triangleq\sum_{r \in\mathcal R}\sum_{f\in\mathcal F}\sum_{v\in\mathcal V}\sum_{t\in\mathcal T}\mathbb E\left(L_{v,r,f,t}\right)\\
	&=\sum_{r \in\mathcal R}\sum_{f\in\mathcal F}\sum_{v\in\mathcal V} \mathbb E\left[\theta_{v,f}^2\right]\gamma_{r,f}\sum_{t\in\mathcal T}\mathbb E\left[\theta_{v,r,t}^1\right],
\end{aligned}
\end{equation}
where $\sum_{t\in\mathcal T}\mathbb E\left[\theta_{v,r,t}^1\right]$ represents the residence time in RSU $r$ of vehicle $v$. From \eqref{total content retrieval delay} and \eqref{Wxy}, one can see that the above total expected caching cost $W\left(\mathbf x,\mathbf y\right)$ only depends on $\mathbf x$ and $\mathbf y$.

\subsection{Problem Formulation}
The caching contents update regularly with a relatively long cycle, e.g., 30 min in \cite{RSUupdate1}, several hours in \cite{RSUupdate2}, and one day in \cite{RSUupdate3}. In this paper, we aim to design a cooperative cache scheme among all RSUs and MBSs in the entire region. Note that since MBSs/RSUs take into account the future served vehicles when they make a decision of caching, the requested contents are deployed by the upcoming RSUs with a high probability so that the intermittency of information transmissions is improved.
The proactive caching problem with the objective of minimizing the total expected caching cost is formulated as follows:
\begin{equation}
\label{Caching Cost Problem}
\begin{aligned}
	& \min_{x_{r,f}, y_{m,f}} && W\left(\mathbf x,\mathbf y\right) \\
	& \,\,\quad\textrm{s.t.} && P_r\left(\mathbf x_r\right)\triangleq\sum_{f\in\mathcal F}x_{r,f} s_f - S_r^{\text{RSU}}\leq 0, \forall r\in\mathcal R, \\
	&&& Q_{m}\left(\mathbf y_m\right)\triangleq\sum_{f\in\mathcal F}y_{m,f} s_f - S_m^\text{MBS}\leq 0, \forall m\in\mathcal M,\\
	&&& x_{r,f}\in\left\{0, 1\right\},\quad y_{m,f}\in\left\{0, 1\right\}.
\end{aligned}
\end{equation}
The caching deployment optimization problem aims to reduce the total expected caching cost via adjusting caching deployment with the limited storage capacity of RSUs and MBSs. Note that as the inherent behavioral properties, the residence time $\sum_{t\in\mathcal T}\mathbb E\left[\theta_{v,r,t}^1\right]$ and the retrieving probability $\mathbb E\left[\theta_{v,f}^2\right]$ are significant for the caching deployment $\mathbf x, \mathbf y$ and the system performance. Meanwhile, they are not affected by (independent of) the caching deployment $\mathbf x, \mathbf y$. Particularly, they will be efficiently estimated in Section \ref{Trajectory Prediction Scheme} and Section \ref{Recommendation System Scheme}, respectively.

Remark: The cloud layer is responsible for collecting the residence time, i.e., $\sum_{t\in\mathcal T}\mathbb E\left[\theta_{v,r,t}^1\right]$ and the probabilities of retrieving the content, i.e.,  $\mathbb E\left[\theta_{v,f}^2\right]$ from MBSs and RSUs and solve the problem (9). Specifically on one hand, RSU can predict the future trajectory of the vehicles, i.e., the residence time $\sum_{t\in\mathcal T}\mathbb E\left[\theta_{v,r,t}^1\right]$, then upload the results to the cloud layer via the local MBS; On the other hand, MBS and RSUs collaboratively execute HFL to predict some contents most likely to be requested, i.e., $\mathbb E\left[\theta_{v,f}^2\right]$, then all MBSs upload the results to the cloud layer. After collecting these information, the cloud layer can solve the problem (9) efficiently.

\section{Trajectory Prediction Scheme}
\label{Trajectory Prediction Scheme}
The residence time of vehicles staying in each RSU is of great importance for content caching placement decisions since the probability of requesting content from RSU $m$ increases in proportional to the time duration going through its communication range. Most previous works assume that the future location can be completely known in advance in some ways, for example, the route can be available by GPS \cite{Huaqing Wu}, or the vehicles are assumed to keep going straight along the expressways \cite{Zhengxin Yu}. However, in practice, the entire GPS data cannot be obtained by all MBSs and RSUs along the road due to privacy concerns. Besides, there are many crossroads and forks making it impossible for vehicles to keep going straight all the time. To compensate this issue, we propose a trajectory prediction scheme in this section.

\subsection{Overall Framework}
In urban roads, there are a large number of intricate types of roads, such as straights, curves, ramps, bridges, tunnels, crosses/T-junctions, etc. Meanwhile, to alleviate the dependence on GPS data and reduce the computational complexity, we are inclined to adopt the connection order of surrounding RSUs to describe the vehicle trajectory, and it is conveniently obtained by recording identifications of all the vehicles served by each RSU.

As shown in Fig. \ref{trajectory}, the overall framework contains historical feature extraction and future trajectory prediction. The historical feature can be extracted from two aspects, namely, daily feature extraction and feature fusion. Specifically, daily features are extracted from the daily trajectory at first, and then fused into the historical feature information. Finally, the future trajectory prediction module aims at predicting the future residence time in each RSU by combining intraday trajectory and historical feature information.
\begin{figure*}[htbp]
\centering
\includegraphics[scale=0.05]{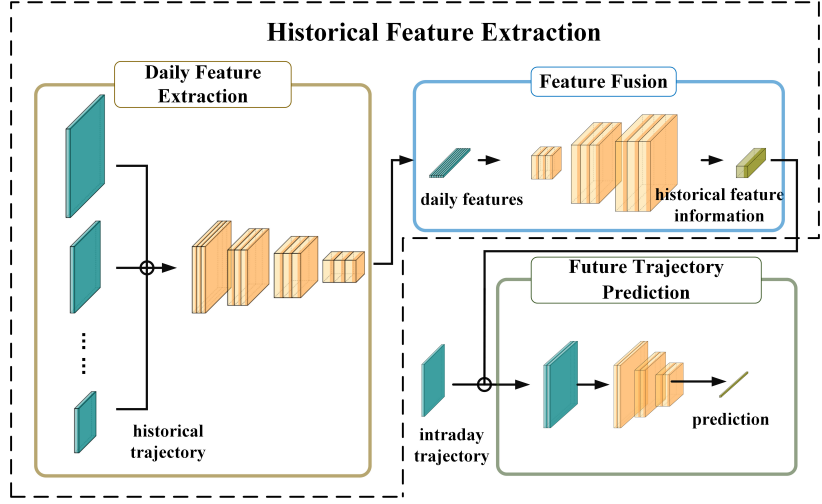}
\caption{Architecture of trajectory prediction model, which consists of historical feature extraction and fusion, as well as future trajectory prediction}
\label{trajectory}
\end{figure*}

\subsection{LSTM-based Trajectory Prediction}
In this section, LSTM-based algorithm \cite{DL} is proposed to extract historical features and make predictions of the future trajectory. In the setting of LSTM-based trajectory prediction model for the vehicle $v$, given its historical trajectory in the last $L$ days $\mathcal Z^v=\left(\mathcal Z_1^v, \mathcal Z_2^v, \cdots, \mathcal Z_{L}^v\right)$ and intraday trajectory $\mathcal Z_{L+1}^v$, we construct an LSTM-based trajectory prediction scheme to map the historical trajectory to its corresponding residence time vector in all the RSUs $\mathbf o^v=\left(o_1^v,o_2^v,\cdots,o_R^v\right)$. 
In this work, for the  target range and the period of time, the location of vehicle will be recorded at every interval. By this means, the trajectory sequence of vehicle $v$ on the $\ell$-th day can be expressed as $\mathcal Z_\ell^v\triangleq\left( z_{\ell,1}^v, z_{\ell,2}^v, \cdots, z_{\ell,N}^v\right)$, whose element $z_{\ell,i}^v\in\mathcal R$ represents the location of vehicle $v$ on the $\ell$-th day at timestamp $i$, and the longer the consecutive identical positions implies the longer time that the vehicle stays in the same RSU. Note that the daily trajectory has a fixed-length of $N$ via truncation or padding. 

\subsubsection{Embedding Layer} For all the $R$ RSUs, we use the zero padding method to create the RSU embedding matrix $\mathbf R\in\mathbb R^{d_\textrm{RSU} \times (R+1)}$, which is a linear map from the RSU set to a $d_\textrm{RSU}$-dimensional vector space. Note that the matrix contains $R+1$ columns since we consider an additional virtual RSU whose element is padded as 0. The embedding matrix for vehicle $v$ in the $\ell$-th day is given by
\begin{equation}
\hat{\mathbf R}_\ell^{v} = \left[
\begin{aligned}
	\mathbf R_{z_{\ell,1}^v},
	\mathbf R_{z_{\ell,2}^v},
	\cdots,
	\mathbf R_{z_{\ell,N}^v}
\end{aligned}
\right],
\end{equation}
where $\mathbf R_j$ is the $j$-th column of the embedding matrix $\mathbf R$.

\subsubsection{Daily Feature Extraction} The information of daily trajectory is extracted by the first LSTM structure \cite[\S 10.10]{DL}, i.e. $\textrm{LSTM}_1$:
\begin{equation}
\mathbf h_{\ell}^v, \mathbf c_{\ell}^v = \textrm{LSTM}_1\left(\hat{\mathbf R}_\ell^{v}\right), \ell=1,2,\cdots, L,
\end{equation}
where $\mathbf h_{\ell}^v, \mathbf c_{\ell}^v$ are the hidden state vector and cell state vector, respectively.

\subsubsection{Feature Fusion} The information of periodic behavioral characteristics is extracted by the second LSTM structure, i.e. $\textrm{LSTM}_2$:
\begin{equation}
\mathbf h_{\textrm{his}}^v, \mathbf c_{\textrm{his}}^v = \textrm{LSTM}_2\left(\mathbf h_{1}^v, \mathbf c_{1}^v, \mathbf h_{2}^v, \mathbf c_{2}^v, \cdots, \mathbf h_{L}^v, \mathbf c_{L}^v\right),
\end{equation}
where $\mathbf h_{\textrm{his}}^v, \mathbf c_{\textrm{his}}^v$ are the final historical feature information.

\subsubsection{Residence Time Prediction} Finally, after extracting the historical information of previous trajectory, the prediction of the residence time is given by
\begin{equation}
\label{TrajectoryPrediction}
\begin{aligned}
\bar{\mathbf h}^v, \bar{\mathbf c}^v &= \textrm{LSTM}_3\left(\hat{\mathbf R}_{L+1}^{v}, \mathbf h_{\textrm{his}}^v, \mathbf c_{\textrm{his}}^v\right),\\
\hat{\mathbf o}^v &= \textrm{ReLU}\left(\bar{\mathbf h}^v \mathbf W + \mathbf b\right),
\end{aligned}
\end{equation}
where $\mathbf W$ and $\mathbf b$ are trainable parameters. ReLU is an activation function defined as $\textrm{ReLU}(x) \triangleq \max\left\{0, x\right\}$. Moreover, all learnable parameters are denoted by $\boldsymbol{\theta}^\textrm{traj}$ including parameters in all LSTM structures and $\mathbf W, \mathbf b$.

In the model training, the input is the sequence $\mathcal Z_\ell^v, \ell=1,2,\cdots,L+1$, the expected output is the corresponding residence time vector $\mathbf o^v$, and the mean squared error loss is adopted as the objective function:
\begin{equation}
\begin{aligned}
	\mathcal L^\textrm{traj}=\sum_{v\in\mathcal V}\frac{1}{2}\|\mathbf o^v-\hat{\mathbf o}^v\|^2.
\end{aligned}
\end{equation}
We adopt offline learning to train the trajectory prediction model which is deployed in all RSUs after training for prediction tasks. Only forward propagation is taken in the prediction stage, so there is no time limit for the training phase. After prediction, the final $\hat{\mathbf o}^v$ is regarded as an estimation of $\sum_{t\in\mathcal T}\mathbb E\left[\theta_{v,r,t}^1\right]$ that can be used in $W\left(\mathbf x, \mathbf y\right)$ in Problem \eqref{Caching Cost Problem}. Note that the residence time can be computed directly.

\section{Recommendation System Scheme}
\label{Recommendation System Scheme}
Since the sequential dynamics are a key feature to capture the context of vehicles' recent activities, in this section, we take the SASRec-based network to predict future content requests of vehicles. Furthermore, with the consideration of the increasing privacy concerns and the ever-growing distributed training data, we propose a Hierarchical Federated learning (HFL) structure to train the SASRec network for each cluster.

\begin{figure*}
\centering
\includegraphics[scale=0.05]{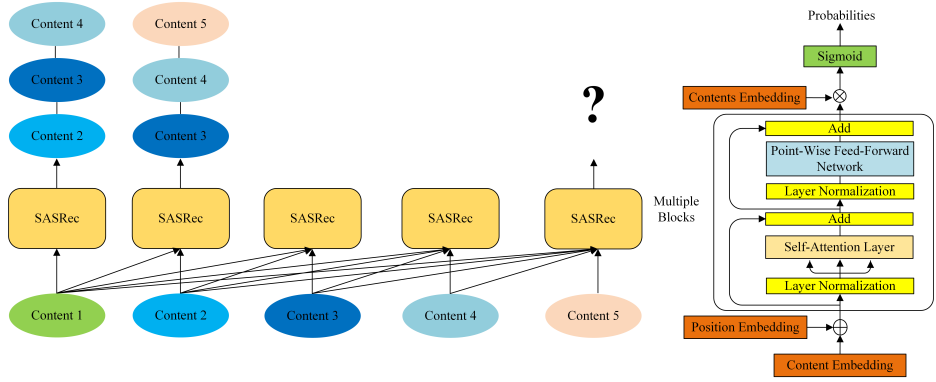}
\caption{(a) A simplified diagram showing the training process of SASRec with $I=5, I^\prime=3$. At each time step, the model considers all previous items to predict the next requested contents. (b) SASRec-model structure.}
\label{recommendation}
\end{figure*}

\subsection{SASRec Model}
As shown in Fig. \ref{recommendation}(a), in the sequential recommendation setting, since the lengths of requested content sequences of each vehicle might be different, it is not desirable to predict the future requirements with all previous contents. We consider a fixed size for all user requested content sequences by truncating or padding, and the maximum length is set as $I$ for the vehicle $v$, i.e., $\mathcal F^v=\left(F_1^v, F_2^v,\cdots,F_{I}^v\right)$. During training, at the $i$-th request file, the model utilizes the previous $i$ files, i.e., $\left(0, 0,\cdots, 0, F_1^v, F_2^v,\cdots,F_{i}^v\right)$, to predict the next $I^\prime>1$ files, where $I-I^\prime-i$ files are padded with $0$. In this paper, we extend the original SASRec model with $I^\prime = 1$ in \cite{SASRec} to $I^\prime>1$, allowing for the recommendation of multiple contents of interest for each vehicle so as to meet the requirements of the vehicle to a great extent.

\subsubsection{Embedding Layer} For the total of $F$ available files, we use the zero padding method to create two embedding matrices $\mathbf M\in\mathbb R^{d \times (F+1)}$ and $\mathbf P\in\mathbb R^{d \times (I-I^\prime)}$ to denote the content embedding matrix and the positional embedding matrix, respectively, where $d$ is the latent dimensionality of both embedding matrices. Note that the content embedding matrix contains $F+1$ columns since we consider an additional virtual file whose element is padded as 0. Then, the embedding matrix for vehicle $v$ is given by
\begin{equation}
\label{step 1}
\hat{\mathbf E}^v = \left[
\begin{aligned}
	\mathbf M_{F_1^v} + \mathbf P_1,
	\mathbf M_{F_2^v} + \mathbf P_2,
	\cdots,
	\mathbf M_{F_{I-I^\prime}^v} + \mathbf P_{I-I^\prime}
\end{aligned}
\right],
\end{equation}
where $\mathbf M_j$ and $\mathbf P_j$ are the $j$-th columns of the embedding matrices $\mathbf M$ and $\mathbf P$, respectively.

\subsubsection{Self-Attention Block} The information of previously consumed contents is extracted by self-attention layer and point-wise feed-forward network. Specifically, the self-attention operation takes the embedding matrix as input, converts it to three matrices through linear projections, and feeds them into an attention layer:
\begin{equation}
\label{step 2}
\begin{aligned}
	&\hat{\mathbf S}^v\triangleq\textrm{SA}\left(\hat{\mathbf E}^v\right)=\textrm{Attention}\left(\hat{\mathbf E}^v\mathbf W^Q, \hat{\mathbf E}^v\mathbf W^K, \hat{\mathbf E}^v\mathbf W^V\right)\\
	&=\textrm{Softmax}\left(\frac{\hat{\mathbf E}^v\mathbf W^Q\left(\hat{\mathbf E}^v\mathbf W^K\right)^T}{\sqrt{d}}\right)\hat{\mathbf E}^v\mathbf W^V,
\end{aligned}
\end{equation}
where  $\mathbf W^Q, \mathbf W^K, \mathbf W^V\in\mathbb{R}^{d\times d}$ denote the projection matrices, $\textrm{Attention}$ denotes the function of scaled dot-product attention mechanism, $d$ is the scale factor to avoid overly large values of the inner product. 
In addition to attention sub-layers, our model contains a fully connected feed-forward network, which is applied to each position separately and identically. This consists of two linear transformations with a ReLU activation in between:
\begin{equation}
\label{step 3}
\begin{aligned}
	\hat{\mathbf F}^v\triangleq\textrm{FFN}\left(\hat{\mathbf S}^v\right) = \textrm{ReLU}\left(\hat{\mathbf S}^v \mathbf W_1 + \mathbf b_1\right)\mathbf W_2 + \mathbf b_2,
\end{aligned}
\end{equation}
where $\mathbf W_1, \mathbf W_2$ are $d \times d$ matrices and $\mathbf b_1, \mathbf b_2$ are $d$ dimensional vectors. 
Besides, three efficient policies can also be considered: i) stacking self-attention blocks to make the model learn more complex content transitions; ii) adopting a dropout operation to avoid overfitting; iii) using residual connections and layer normalization to stabilize and accelerate the network training process. 

\subsubsection{Output Layer}
Finally, after multiple self-attention blocks extract information of previously requested contents, the prediction of the next contents is given by
\begin{equation}
\label{step 4}
\hat{r}_{i,f}^v = \textrm{Sigmoid}\left(\hat{\mathbf F}_i^v \mathbf M_f \right),
\end{equation}
where $\hat{\mathbf F}_i^v$ is the $i$-th row of the matrix $\hat{\mathbf F}^v$ and denotes the feature vector of the vehicle $v$ after the $i$-th request, $\mathbf M_f$ is the $f$-th row of matrix $\mathbf M$ and denotes the item embedding vector of content $f$. The goal of original SASRec model \cite{SASRec} is to seek to identify which items are `relevant' from one user's historical behavior and use them to predict the next item. Limited by the performance of predicting the next one item accurately, it is not effective to predict the future consecutive multiple items with recursive multi-step forecast method. As shown in Fig. 4(b), the redesigned SASRec model adds a Sigmoid layer to represent the probability of a vehicle requesting all contents over a period of time in the future.

In the network training, let $\boldsymbol{\theta}^\textrm{rec}=\{\mathbf{M}, \mathbf P, \mathbf W^Q, \mathbf W^K, \mathbf W^V, \mathbf W_1,  \mathbf W_2, \mathbf b_1, \mathbf b_2\}$  denote all the learnable parameters, and we adopt the following binary cross entropy loss as the objective function:
\begin{equation}
\begin{aligned}
&\mathcal L^\textrm{rec} \triangleq\sum_{v\in\mathcal V} \mathcal L_v^\textrm{rec}  \\
&\triangleq \sum_{v\in\mathcal V}\sum_{i=1}^I\left[\sum_{f\in\tilde{\mathcal F}_{i}^v}\log\left(\hat{r}_{i,f}^v\right)+\sum_{f\notin\tilde{\mathcal F}^v}\log\left(1-\hat{r}_{i,f}^v\right)\right].
\end{aligned}
\end{equation}
Different from trajectory prediction in the last section, the training of recommendation system by minimizing $\mathcal L^\textrm{rec}$ can only be executed when the vehicle stays in the coverage of one certain RSU. The distributed training method will be introduced in the next subsection. After training, the estimation of $\mathbb E\left[\theta_{v,f}^2\right]$ for the next requested items can be computed via the same steps as \eqref{step 1}-\eqref{step 4}, which can be used in $W\left(\mathbf x, \mathbf y\right)$ in Problem \eqref{Caching Cost Problem}.

\subsection{HFL-based SASRec System}
In order to protect the privacy of user data, we take an HFL-based structure to train the SASRec systems instead of the centralized training in \cite{SASRec}. Compared with the traditional federated learning designed for a single cluster, HFL architecture is more suitable for more datasets provided by massive vehicles in a larger network, which can also improve the accuracy of the model. On the other hand, the vehicle stays in the MBS for a longer time, which helps to have more time to stabilize the model's performance. Thus, as shown in Fig. \ref{HFL}, we consider an HFL system that has one MBS $m$, $R^m$ RSUs indexed by $\mathcal R_m = \left\{1,2,\cdots, R^m\right\}$, and $V^m$ vehicles indexed by $\mathcal V_m = \left\{1,2,\cdots, V^m\right\}$. RSU $r\in\mathcal R_m$ manages $V_r^m$ vehicles indexed by $\mathcal V_r^m=\left\{1,2,\cdots, V_r^m\right\}$.

The key steps of the HFL proceed as follows. After every $\kappa_1$ local updates at each vehicle, each RSU will collect and aggregate local models from its vehicles, and then distribute  the aggregated model to them. After every $\kappa_2$ edge model aggregations at each RSU, the MBS will collect and aggregate all edge models from all RSUs in its cluster,  and distribute the latest global model to them, which means the global model aggregations at MBSs happen every $\kappa_1\kappa_2$ local updates. Let $\boldsymbol{\theta}_{v}^\textrm{rec}\left(\kappa\right)$ denote the local model parameters of vehicle $v\in\mathcal V_r^m$ after the $\kappa$-th local update. The evolution of local model parameters $\boldsymbol{\theta}_{v}^\textrm{rec}\left(\kappa\right)$ is given by
\begin{equation}
\boldsymbol{\theta}_{v}^\textrm{rec}\left(\kappa\right)=\left\{
\begin{aligned}
&\boldsymbol{\theta}_{v}^\textrm{rec}\left(\kappa-1\right) - \eta^\textrm{rec}\nabla\mathcal L_v^\textrm{rec}, & \kappa | \kappa_1 \neq 0\\
&\frac{\sum_{v\in\mathcal V_r^m} \boldsymbol{\theta}_{v}^\textrm{rec}\left(\kappa-1\right)}{V_r^m}, & \kappa | \kappa_1 = 0, \kappa | \kappa_1\kappa_2 \neq 0\\
&\frac{\sum_{v\in\mathcal V^m} \boldsymbol{\theta}_{v}^\textrm{rec}\left(\kappa-1\right)}{V^m}, & \kappa | \kappa_1\kappa_2 = 0
\end{aligned}\right.
\end{equation}
where $\eta^\textrm{rec}$ is the step size of gradient descent. Different from offline training for trajectory prediction model, the training for recommendation system is time-sensitive due to the high dynamics of vehicles. All vehicles continuously interact with all RSUs and MBSs along the road to transmit gradients other than raw historical user preferences, so as to protect user privacy. Once a vehicle leaves the current RSU before accomplish local training, it fails to upload the updated model to the target RSU, 
which might  lead to waste of  computation capacity. The authors in \cite{Zhengxin Yu} propose  a simple measure  by selecting slow-moving vehicles that can finish local training and uploading before leaving as the participants in edge model training. The details of HFL-based SASRec system are presented in Algorithm \ref{ContentPrediction}. Note that in the HFL-based SASRec system, the convergence is guaranteed in \cite[\S III]{HFL}.

In our future work, we will further improve the performance of the content popularity prediction scheme to better adapt to the characteristics of vehicular networks, including establishing more effective prediction models for special areas and developing faster training schemes.
\begin{algorithm}
\caption{HFL-based SASRec}
\label{ContentPrediction}
{\bf Input:} Initial model parameter $\boldsymbol{\theta}_{v}^\textrm{rec}\left(0\right)$. {\bf Output:} $\left\{\mathbb E\left[\theta_{v,f}^2\right]\right\}_{v\in\mathcal V,f\in\mathcal F}$

\begin{algorithmic}[1]
	\FOR{$\kappa=1, 2, ...$}
	\FOR{each cluster in parallel}
	\IF{$\kappa | \kappa_1 \neq 0$} 
	\STATE Each vehicle updates its local model with $\boldsymbol{\theta}_{v}^\textrm{rec}\left(\kappa\right)=\boldsymbol{\theta}_{v}^\textrm{rec}\left(\kappa-1\right) - \eta\nabla\mathcal L_v^\textrm{rec}$ in parallel. 
	\ELSIF{$\kappa | \kappa_1\kappa_2 \neq 0$}
	\STATE Each RSU in cluster $m$ aggregates models: $\boldsymbol{\theta}_{v}^\textrm{rec}\left(\kappa\right)=\frac{\sum_{v\in\mathcal V_r^m} \boldsymbol{\theta}_{v}^\textrm{rec}\left(\kappa-1\right)}{V_r^m}$ and downloads to all served vehicles in parallel.
	\ELSE
	\STATE MBS $m$ aggregates models: $\boldsymbol{\theta}_{v}^\textrm{rec}\left(\kappa\right)=\frac{\sum_{v\in\mathcal V^m} \boldsymbol{\theta}_{v}^\textrm{rec}\left(\kappa-1\right)}{V^m}$ and downloads to all vehicles.
	\ENDIF
	\ENDFOR
	\ENDFOR
	\STATE Estimate the probability of vehicles $v$ retrieving content $f$, i.e., $\mathbb E\left[\theta_{v,f}^2\right] = \textrm{Sigmoid}\left(\hat{r}_{I,f}^v\right)$.
\end{algorithmic}
\end{algorithm}

\begin{figure*}[htbp]
\centering
\includegraphics[scale=0.13]{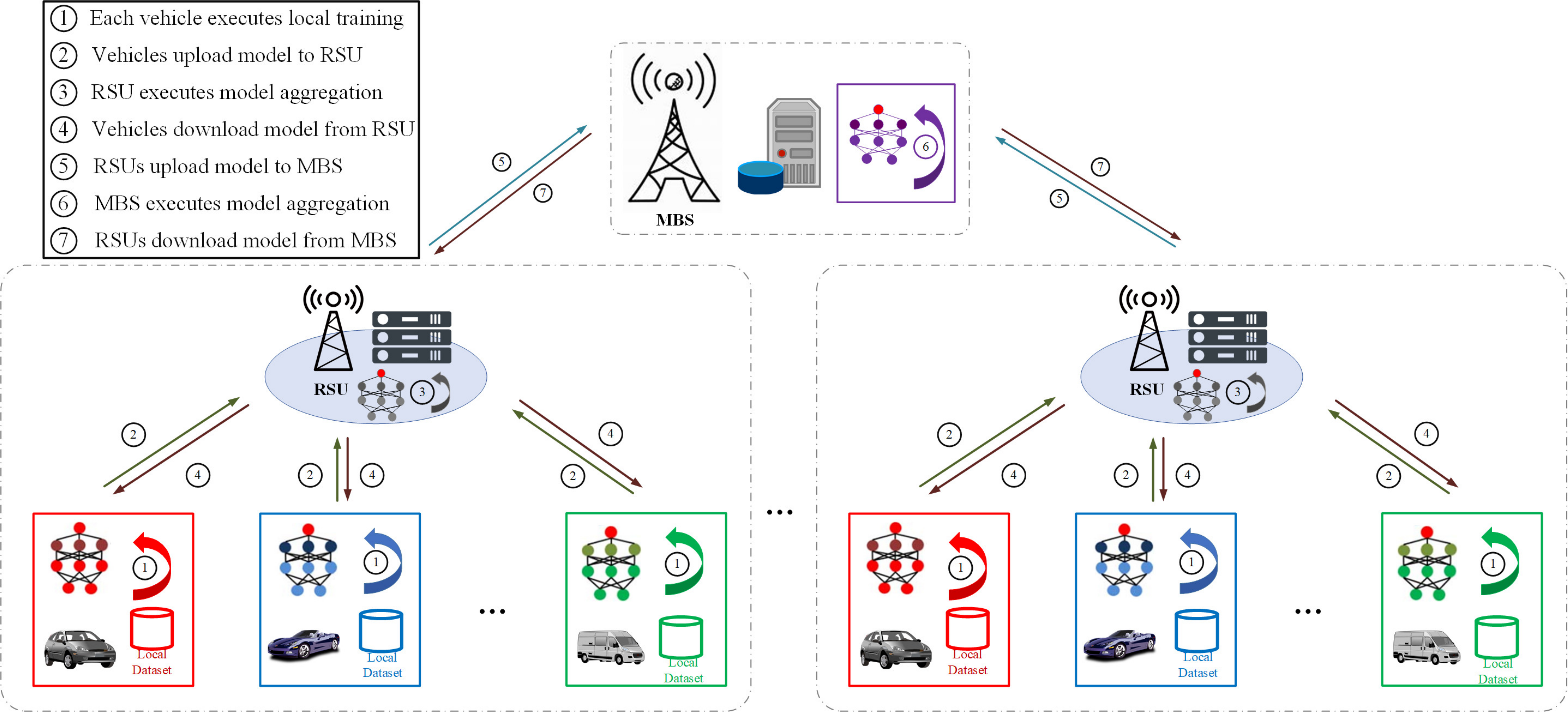}
\caption{Framework of the HFL-based SASRec system.}
\label{HFL}
\end{figure*}

\section{Dynamic Content Caching Scheme}
\label{section: dynamic content caching scheme}
In this section, we propose an adaptive gradient descent-based content caching policy to tackle the large-scale 0-1 constrained problem, including  continuous relaxation and penalty coefficient adaptation. Moreover, dynamic content caching is to integrate the proposed caching policy with the aforementioned trajectory prediction and content popularity prediction for the dynamic topology and the real-time request.

\subsection{Penalty Method}
Solving the original problem \eqref{Caching Cost Problem} faces two main challenges: i) the number of binary decision variables $x_{r,f}$ and $y_{m,f}$ are tremendous due to abundant contents; and ii) the number of constraints are enormous due to the high density of edge nodes. Both issues severely hinder the efficacy of tackling this large-scale optimization problem. Inspired by \cite{0-1 KP}, we first relax $(M+R)F$ constraints about binary decision variables to
\begin{equation}
\begin{aligned}
	0\leq &x_{r,f}\leq 1, \quad 0\leq &y_{m,f}\leq 1.
\end{aligned}
\end{equation}
Since tackling the original caching problem with numerous constraints that may greatly exceed nowaday computing capability, the Sigmoid function is introduced to remove these constraints as: $x_{r,f} = \textrm{Sigmoid}(\tilde{x}_{r,f}), \tilde{x}_{r,f}\in\mathbb R$ and $y_{m,f} = \textrm{Sigmoid}(\tilde{y}_{m,f}), \tilde{y}_{m,f}\in\mathbb R$. More formally, the binary decision variables are relaxed as follows:
\begin{equation}
\label{Recover1}
\begin{aligned}
x_{r,f} & \triangleq h\left(\tilde{x}_{r,f}\right) = \textrm{Sigmoid}\left(\tilde{x}_{r,f}\right),\\
y_{m,f} & \triangleq h\left(\tilde{y}_{m,f}\right) = \textrm{Sigmoid}\left(\tilde{y}_{m,f}\right).
\end{aligned}
\end{equation}
With the proposed approximation, the original problem can be reformulated as follows:
\begin{equation}
\label{relaxation}
\begin{aligned}
\min_{\tilde{x}_{r,f}, \tilde{y}_{m,f}} &W\left(h\left(\mathbf{\tilde{x}}\right),h\left(\mathbf{\tilde{y}}\right)\right) \\
\textrm{s.t.} \quad & P_r\left(h\left(\tilde{\mathbf x}_r\right)\right)\leq 0, \quad\forall r\in\mathcal R, \\
&Q_{m}\left(h\left(\mathbf{\tilde{y}}_m\right)\right)\leq 0, \quad\forall m\in\mathcal M.
\end{aligned}
\end{equation}
This relaxation reduces the number of constraints from $(M+R)(F+1)$ to $M+R$. For ease of notations in this section, we simply use $W\left(\mathbf{\tilde{x}},\mathbf{\tilde{y}}\right)$, $P_r\left(\tilde{\mathbf x}_r\right)$, and $Q_{m}\left(\mathbf{\tilde{y}}_m\right)$ to represent $W\left(h\left(\mathbf{\tilde{x}}\right),h\left(\mathbf{\tilde{y}}\right)\right)$, $P_r\left(h\left(\tilde{\mathbf x}_r\right)\right)$, and $Q_{m}\left(h\left(\mathbf{\tilde{y}}_m\right)\right)$, respectively, and use $\mathbf w=\left(\mathbf w_{\mathbf{\tilde{x}}_r}, \mathbf w_{\mathbf{\tilde{y}}_r}\right)$,  $\mathbf p_r$, and $\mathbf q_m$ to represent their gradients w.r.t. $\mathbf{\tilde{x}},\mathbf{\tilde{y}}$, respectively.

The penalty method  in \cite[Eq. 17.2]{NumOptim} enables us to rewrite \eqref{relaxation} as follows:
\begin{equation}
\label{extended objective function}
\begin{aligned}
\min_{\mathbf{\tilde{x}},\mathbf{\tilde{y}}}&\quad L\triangleq W\left(\mathbf{\tilde{x}},\mathbf{\tilde{y}}\right)+\frac{1}{2}\beta\sum_{r\in\mathcal R} \textrm{ReLU}\left[P_r\left(\tilde{\mathbf x}_r\right)\right]^2 \\
&\qquad+\frac{1}{2}\beta\sum_{m\in\mathcal M} \textrm{ReLU}\left[Q_{m}\left(\mathbf{\tilde{y}}_m\right)\right]^2,
\end{aligned}
\end{equation}
where $L$ is the extended objective function and $\beta >0 $ is the penalty coefficient. We utilize the gradient descent method to optimize $L$ by
\begin{equation}
\label{gradient update}
\begin{aligned}
	\mathbf{\tilde{x}}_r&\leftarrow \mathbf{\tilde{x}}_r -\eta\left(\nabla L\right)_{\mathbf{\tilde{x}}_r},\quad
	\mathbf{\tilde{y}}_m&\leftarrow \mathbf{\tilde{y}}_m -\eta\left(\nabla L\right)_{\mathbf{\tilde{y}}_m},
\end{aligned}
\end{equation}
where $\eta>0$ is  a sufficiently small stepsize (or say learning rate). $\nabla L$ is the gradient of $L$ and its $\mathbf{\tilde{x}}_r$-component and $\mathbf{\tilde{y}}_m$-component are given by
\begin{equation}
\begin{aligned}
\left(\nabla L\right)_{\mathbf{\tilde{x}}_r}  &= \mathbf w_{\mathbf{\tilde{x}}_r} + \beta \textrm{ReLU}\left[P_r\left(\tilde{\mathbf x}_r\right)\right] \mathbf p_r,\\
\left(\nabla L\right)_{\mathbf{\tilde{y}}_m}  &= \mathbf w_{\mathbf{\tilde{y}}_m} + \beta \textrm{ReLU}\left[Q_m\left(\tilde{\mathbf y}_m\right)\right] \mathbf q_m.
\end{aligned}
\end{equation}

The role of $\beta$ is essential for the convergence of the optimization process. 
From experience, a large $\beta$ is preferred for a fast convergence rate, however, an  overlarge $\beta$ will
disturb the gradient descent algorithm once faced a small cost overrun. Therefore, we propose the following adaptive gradient descent algorithm that dynamically controls $\beta$ in the next subsection, where the coefficient is adjusted based on information from the solution obtained at last iteration.

\subsection{Adaptive Gradient Descent Algorithm}
We first introduce the following preliminary result that the objective function and constraints have Lipschitz continuous gradients.
\begin{Proposition}
\label{prop0}
The gradients of the objective function and the constraints are all Lipschitz continuous, i.e.,
\begin{equation}
	\begin{aligned}
		& \|\nabla W\left(\mathbf{\tilde{x}_1},\mathbf{\tilde{y}_1}\right) - \nabla W\left(\mathbf{\tilde{x}_2},\mathbf{\tilde{y}_2}\right)\| \leq \lambda_w \|(\mathbf{\tilde{x}_1}, \mathbf{\tilde{y}_1}) -(\mathbf{\tilde{x}_2},\mathbf{\tilde{y}_2})\|,\\
		& \|\nabla P_r\left(\mathbf{\tilde{x}_{r,1}}\right) - \nabla P_r\left(\mathbf{\tilde{x}_{r,2}}\right)\|\leq \lambda_r \|\mathbf{\tilde{x}_{r,1}} - \mathbf{\tilde{x}_{r,2}}\|,\\
		& \|\nabla Q_m\left(\mathbf{\tilde{y}_{m,1}}\right) - \nabla Q_m\left(\mathbf{\tilde{y}_{m,2}}\right)\|\leq \lambda_m \|\mathbf{\tilde{y}_{m,1}} - \mathbf{\tilde{y}_{m,2}}\|.\\
	\end{aligned}
\end{equation}
\end{Proposition}
\begin{proof}
See Appendix \ref{Appendix_prop0}.
\end{proof}

The key idea of the adaptive gradient descent algorithm is to balance all the descent directions of the objective function and the budget constraints, which ensures the objective function continue to decline under the premise of guaranteed constraints. Specifically, when $P_r(\tilde{\mathbf x}_r) \leq 0, Q_m(\tilde{\mathbf y}_m) \leq 0$, i.e., no cost overrun, there is no need to control $\beta$ thanks to the rectifier ReLU. Otherwise, we need simply to set $\beta = 0$ to enable a more aggressive search for the objective without considering the budget constraint. When $P_r(\tilde{\mathbf x}_r) > 0, Q_m(\tilde{\mathbf y}_m) > 0$, we need to adjust $\eta$ and $\beta$ for its steepest gradient descent direction. Next, we give some conditions about $\eta$ and $\beta$ to decrease the objective and constraints.

\begin{Proposition}
\label{prop1}
The value of objective function does not increase after one gradient descent update, i.e., 
$ W\left(\mathbf{\tilde{x}}-\eta\left(\nabla L\right)_{\mathbf{\tilde{x}}},\mathbf{\tilde{y}} -\eta\left(\nabla L\right)_{\mathbf{\tilde{y}}} \right) \leq W\left(\mathbf{\tilde{x}},\mathbf{\tilde{y}}\right),$
if the following sufficient conditions hold:
\begin{equation}
	\label{obj_eta}
	\begin{aligned}
		0\leq\eta\leq \eta_w \triangleq \frac{2\left(\|\mathbf w\|^2+\beta\phi\right)}{\lambda_w\left(\|\left(\nabla L\right)_{\mathbf{\tilde{x}}_r}\|^2 + \|\left(\nabla L\right)_{\mathbf{\tilde{y}}_m}\|^2\right)},
	\end{aligned}
\end{equation}

\begin{equation}
	\label{obj_beta}
	\beta\left\{
	\begin{aligned}
		&\geq 0, &\quad \phi\geq 0,\\
		&\leq-\frac{\|\mathbf w\|^2}{\phi}, &\quad \phi< 0,
	\end{aligned}\right.
\end{equation}
where $\phi$ is defined as $\phi \triangleq \sum_{r\in\mathcal R} \textrm{ReLU}\left[P_r\left(\tilde{\mathbf x}_r\right)\right] \mathbf p_r^T \mathbf w_{\mathbf{\tilde{x}}_r} + \sum_{m\in\mathcal M} \textrm{ReLU}\left[Q_{m}\left(\mathbf{\tilde{y}}_m\right)\right] \mathbf q_{m}^T \mathbf w_{\mathbf{\tilde{y}}_m}$.
\end{Proposition}

\begin{proof}
See Appendix \ref{Appendix_prop1}.
\end{proof}

\begin{Proposition}
\label{prop2}
The caching constraint in RSU, i.e., $P_r\left(\tilde{\mathbf x}_r\right)>0$, does not increase after one gradient descent update, i.e. $
P_r\left(\tilde{\mathbf x}_r-\eta\left(\nabla L\right)_{\mathbf{\tilde{x}}_r}\right)<P_r\left(\tilde{\mathbf x}_r\right)$
if the sufficient condition holds:

\begin{subequations}
	\begin{align}
		& 0\leq\eta\leq\frac{2\left(\mathbf w_{\mathbf{\tilde{x}}_r}^T\mathbf p_r + \beta \mathrm{ReLU}\left[P_r\left(\tilde{\mathbf x}_r\right)\right] \|\mathbf p_r\|^2\right)}{\lambda_r\|\left(\nabla L\right)_{\mathbf{\tilde{x}}_r}\|^2}\triangleq\eta_r^1,\label{RSU_eta}\\
		& \beta\geq-\frac{\mathbf w_{\mathbf{\tilde{x}}_r}^T\mathbf p_r}{\mathrm{ReLU}\left[P_r\left(\tilde{\mathbf x}_r\right)\right] \|\mathbf p_r\|^2}\triangleq\beta_r^{\textrm{RSU}}\label{RSU_beta}.
	\end{align}
\end{subequations}

On the other hand, the caching constraint in RSU, i.e., $P_r\left(\tilde{\mathbf x}_r\right)\leq0$, still holds after one gradient descent update, i.e. $
P_r\left(\tilde{\mathbf x}_r-\eta\mathbf w_{\mathbf{\tilde{x}}_r}\right)\leq 0$ and $\mathbf w_{\mathbf{\tilde{x}}_r} \neq \mathbf 0$, if the following sufficient condition holds:
\begin{equation}
\label{RSU_eta_beta}
\begin{aligned}
& 0 \leq \eta \leq \frac{\mathbf w_{\mathbf{\tilde{x}}_r}^T\mathbf p_r+\sqrt{\left(\mathbf w_{\mathbf{\tilde{x}}_r}^T\mathbf p_r\right)^2-2\lambda_r\|\mathbf w_{\mathbf{\tilde{x}}_r}\|^2P_r\left(\tilde{\mathbf x}_r\right)}}{\lambda_r\|\mathbf w_{\mathbf{\tilde{x}}_r}\|^2}\triangleq\eta_r^2, \\
& \beta = 0.
\end{aligned}
\end{equation}
\end{Proposition}

\begin{proof}
See Appendix \ref{Appendix_prop2}.
\end{proof}

Similarly, to decrease $Q_m\left(\tilde{\mathbf y}_m\right)$ when the constraint in MBS is violated, $\beta$ and $\eta$ should be adjusted as follows
\begin{subequations}
\begin{align}
	& 0\leq\eta\leq\frac{2\left(\mathbf w_{\mathbf{\tilde{y}}_m}^T\mathbf q_m + \beta \mathrm{ReLU}\left[Q_m\left(\tilde{\mathbf y}_m\right)\right] \|\mathbf q_m\|^2\right)}{\lambda_r\|\left(\nabla L\right)_{\mathbf{\tilde{y}}_m}\|^2}\triangleq\eta_m^1,\label{MBS_eta}\\
	& \beta \geq -\frac{\mathbf w_{\mathbf{\tilde{y}}_m}^T\mathbf q_m}{\textrm{ReLU}\left[Q_m\left(\tilde{\mathbf y}_m\right)\right] \|\mathbf q_m\|^2}\triangleq\beta_m^{\textrm{MBS}}\label{MBS_beta}.
\end{align}
\end{subequations}
When $Q_m\left(\tilde{\mathbf y}_m\right)\leq0$, the caching constraint still holds after one gradient descent update, i.e. $Q_m\left(\tilde{\mathbf y}_m-\eta\mathbf w_{\mathbf{\tilde{y}}_m}\right)\leq 0$, if $\eta$ and $\beta$ satisfy the following constraints: 
\begin{equation}
\label{MBS_eta_beta}
\begin{aligned}
	&0 \leq \eta \leq \frac{\mathbf w_{\mathbf{\tilde{y}}_m}^T\mathbf q_m+\sqrt{\left(\mathbf w_{\mathbf{\tilde{y}}_m}^T\mathbf q_m\right)^2-2\lambda_r\|\mathbf w_{\mathbf{\tilde{y}}_m}\|^2Q_m\left(\tilde{\mathbf y}_m\right)}}{\lambda_m\|\mathbf w_{\mathbf{\tilde{y}}_m}\|^2}\\
 &\qquad\triangleq\eta_m^2,\\
	&  \beta = 0.
\end{aligned}
\end{equation}

Note that if all constraints are satisfied, $\beta$ and $\eta$ should be set as follows  to decrease the objective function: 
\begin{equation}
\label{eta_beta_final1}
\begin{aligned}
	\beta = 0,\quad
	0\leq\eta\leq\min\left\{\eta_r^2, \eta_m^2\right\}.
\end{aligned}
\end{equation}

If all of \eqref{obj_beta}, \eqref{RSU_beta}, and \eqref{MBS_beta} can be achieved, while other constraints are violated, the best outcome can be obtained by setting $\beta$ and $\eta$ as follows to decrease the objective function and constraints at the same time: 

\begin{subequations}
\label{eta_beta_final2}
\begin{equation}
\beta\left\{
\begin{aligned}
& \geq \max\left\{0, \beta_r^\textrm{RSU}, \beta_m^\textrm{MBS}\right\}, & \phi\geq 0,\\
& \in \left[\max\left\{\beta_r^\textrm{RSU}, \beta_m^\textrm{MBS}\right\}, -\frac{\|\mathbf w\|^2}{\phi}\right], &\\
&\qquad\quad\phi< 0 \& -\frac{\|\mathbf w\|^2}{\phi}>\beta_r^\textrm{RSU}, \beta_m^\textrm{MBS},
\end{aligned}\right.
\end{equation}

\begin{equation}
	\begin{aligned}
		0\leq\eta\leq\min\left\{\eta_w, \eta_r^1, \eta_m^1\right\},
	\end{aligned}
\end{equation}
\end{subequations}

However, the non-increase in the objective function and the decrease in the constraints may not be fulfilled at the same time. Considering these conditions with some violated constraints, after one gradient descent update, $\beta$ and $\eta$ should be set as follows to satisfy the constraints:
\begin{equation}
\label{eta_beta_final3}
\begin{aligned}
\beta&\geq \max\left\{\beta_r^\textrm{RSU} + \epsilon_r, \beta_m^\textrm{MBS} + \epsilon_m\right\},\\
\eta&>0 \textrm{ is sufficiently small},
\end{aligned}
\end{equation}
where $\epsilon_r$ and $\epsilon_m$ are positive values that can analytically decided by the following Prop. \ref{prop3}. 

\begin{Proposition}
\label{prop3}
There will be no cost overrun, i.e. $P_r\left(\tilde{\mathbf x}_r\right)>0, Q_m\left(\tilde{\mathbf y}_m\right)>0$, after one gradient ascent update if
\begin{equation}
	\label{epsilon_r}
	\epsilon_r=\frac{1}{\eta\|\mathbf p_r\|^2}, \quad\quad\epsilon_m=\frac{1}{\eta\|\mathbf p_m\|^2}.
\end{equation}
\end{Proposition}
\begin{proof}
See Appendix \ref{Appendix_prop3}.
\end{proof}

Finally, We present the convergence analysis and computational complexity analysis on our proposed adaptive gradient descent algorithm.

\textit{Convergence Analysis:} The convergence is guaranteed by the following two facts. First, the objective value of Problem \eqref{relaxation} is non-increasing over iterations, even with arbitrarily initialization, all the constraints in Problem \eqref{relaxation} can be satisfied by tuning $\beta$ and $\eta$. Second, the optimal value of problem \eqref{relaxation} is bounded from below due to the cache constraint. Thus, the objective value is guaranteed to converge. Furthermore, since each binary decision variable is bounded, thus there must exist a convergent subsequence.

\textit{Complexity Analysis:} The complexity of updating $\mathbf x$ and $\mathbf y$ mainly depends on the computation of $\beta$ and the optimization of (24) with gradient descent method. Firstly, it is necessary to compute the penalty coefficient $\beta$ with gradient computation and matrix multiplication according to (30b) and (32b), whose complexity is $\mathcal O\left(F(R+M)\right)$. Secondly, we need to optimize the extended objective function with gradient descent method, whose complexity is $\mathcal O\left(F(R+M)\right)$. Suppose that the proposed algorithm requires $T$ iterations to converge in total. Therefore, the complexity of evaluating $\mathbf x$ and $\mathbf y$ is $\mathcal O(TF(R+M))$.

\subsection{Practical Adaptive Gradient Descent Algorithm}
Although the convergence can be guaranteed as stated in the last subsection, some parameters are difficult to calculate accurately, especially for all the Lipschitz constants. If the stepsize is too small, the convergence rate will be very slow, which is unadaptive to the characteristics of rapid changes in the vehicular networks. In this section, we propose a practical scheme to optimize cooperative cache problem. Eqs. \eqref{eta_beta_final1}, \eqref{eta_beta_final2}, and \eqref{eta_beta_final3} offer a fresh insight into the update of $\beta$, and we give one realization as follows: 

\begin{equation}
\label{beta_final}
\beta=\left\{
\begin{aligned}
&\max\left\{0, \beta_r^\textrm{RSU}, \beta_m^\textrm{MBS}\right\}, &\quad \phi\geq 0,\\
&-\frac{\|\mathbf w\|^2}{2\phi}+\frac{1}{2}\max\left\{\beta_r^\textrm{RSU}, \beta_m^\textrm{MBS}\right\}, &\\
&\quad\quad\quad\phi< 0 \& -\frac{\|\mathbf w\|^2}{\phi}>\beta_r^\textrm{RSU}, \beta_m^\textrm{MBS},\\
&\max\left\{\beta_r^\textrm{RSU} + \epsilon_r, \beta_m^\textrm{MBS} + \epsilon_m\right\}, &\textrm{otherwise}.
\end{aligned}\right.
\end{equation}

Besides, for the purpose of reducing computational complexity and accelerating convergence, we take a slightly larger constant stepsize $\eta$. Note that when there is only one constraint, our realization of $\beta$ is same as the scheme in \cite[Eq. 6]{0-1 KP}.

\subsection{Dynamic Content Caching}
For adapting to the dynamic topology and the real-time request, dynamic content caching mainly integrates the proposed adaptive gradient descent-based caching policy with the trajectory prediction and content popularity predictions. The overall algorithm is outlined in Algorithm \ref{Dynamic Caching Algorithm}. Generally speaking, prediction and caching are executed periodically. In Step 2, RSUs first recognize all vehicles in their coverage, and then predict the future trajectory with the proposed algorithm in Section \ref{Trajectory Prediction Scheme}. In Step 3, MBSs and RSUs collaboratively execute the HFL-based SASRec algorithm proposed in Section \ref{Recommendation System Scheme} to predict the future content requests for each vehicles. Step 4 combines trajectory prediction with content prediction to get the content popularity for each RSU. Steps 5-11 optimize the target caching problem by relaxing it to an approximate problem. Specifically, Steps 6-7 select a suitable $\beta$ to balance all the descent directions; Steps 8-9 compute gradient and update the caching. Finally, Step 12 aims to cache contents for each RSU and MBS according to the caching decisions.

\begin{algorithm}
\caption{Dynamic Caching Algorithm}
\label{Dynamic Caching Algorithm}
{\bf Input:} Stepsize $\eta$, initial $\beta=0$.
\begin{algorithmic}[1]
	\FOR{episode=1, 2, ...}
	\STATE \underline{\textit{Prediction}}: Predict the time of each vehicle entering each RSU $\mathbf P_1$ in Eq. \eqref{TrajectoryPrediction}.
	
	\STATE \underline{\textit{Prediction}}: Predict the probability of each vehicle requesting each content $\mathbf P_2$ in Alg. \ref{ContentPrediction}.
	
	\STATE Compute the content popularity of each RSU, $\mathbf P=\mathbf P_1^T\mathbf P_2$.
	
	\REPEAT
	\STATE Compute $\beta_r^\textrm{RSU}, \beta_m^\textrm{MBS}$ with \eqref{RSU_beta}, \eqref{MBS_beta} when constraints have been violated.
	
	\STATE Compute $\beta$ with \eqref{beta_final}.
	
	\STATE Compute gradient $\nabla L$ of extend objective function in \eqref{extended objective function} with given $\beta$.
	
	\STATE Update the optimization variables $\tilde{\mathbf x}, \tilde{\mathbf y}$ with \eqref{gradient update}.
	\UNTIL{Convergence}
	
	\STATE \underline{\textit{Caching}}: Sort $\tilde{\mathbf x}, \tilde{\mathbf y}$, and cache contents in turn until the constraint is violated.
	\ENDFOR
\end{algorithmic}
\end{algorithm}

\section{Results and Discussions}
\label{Simulation}
\subsection{Simulation Setup}
In this section, we numerically evaluate the performance of the proposed proactive content caching scheme. In this simulation, there are 1107 vehicles served by the hierarchical cooperative caching network, the cached contents of RSUs and MBSs are determined based on the prediction of the vehicle mobility and content preference. As shown in Fig. \ref{simulation_setup}, the map is roughly $3km\times 3km$ range in Shenzhen. For the purpose of covering the main streets entirely, the urban area is divided into 8 clusters, each of which consists of 1 MBS and 10 RSUs. The mobility trajectory over the road network is generated by SUMO simulator to imitate behaviors of vehicles \cite{SUMO}. The dataset of the recommendation system used in our experiments is MovieLens 1M dataset collected from the MovieLens website \cite{Movielens}. About 1 million ratings are contained in this dataset, which came from 6040 anonymized users on 3416 movies. To simulate the process of vehicles' requests, the rated movies are assumed as request contents from vehicles. The system parameters used for our simulations are listed in Table \ref{Simulation Parameters}.

\begin{figure}[!htb]
\begin{minipage}[b]{0.5\textwidth}
	\centering
	\includegraphics[scale=0.12]{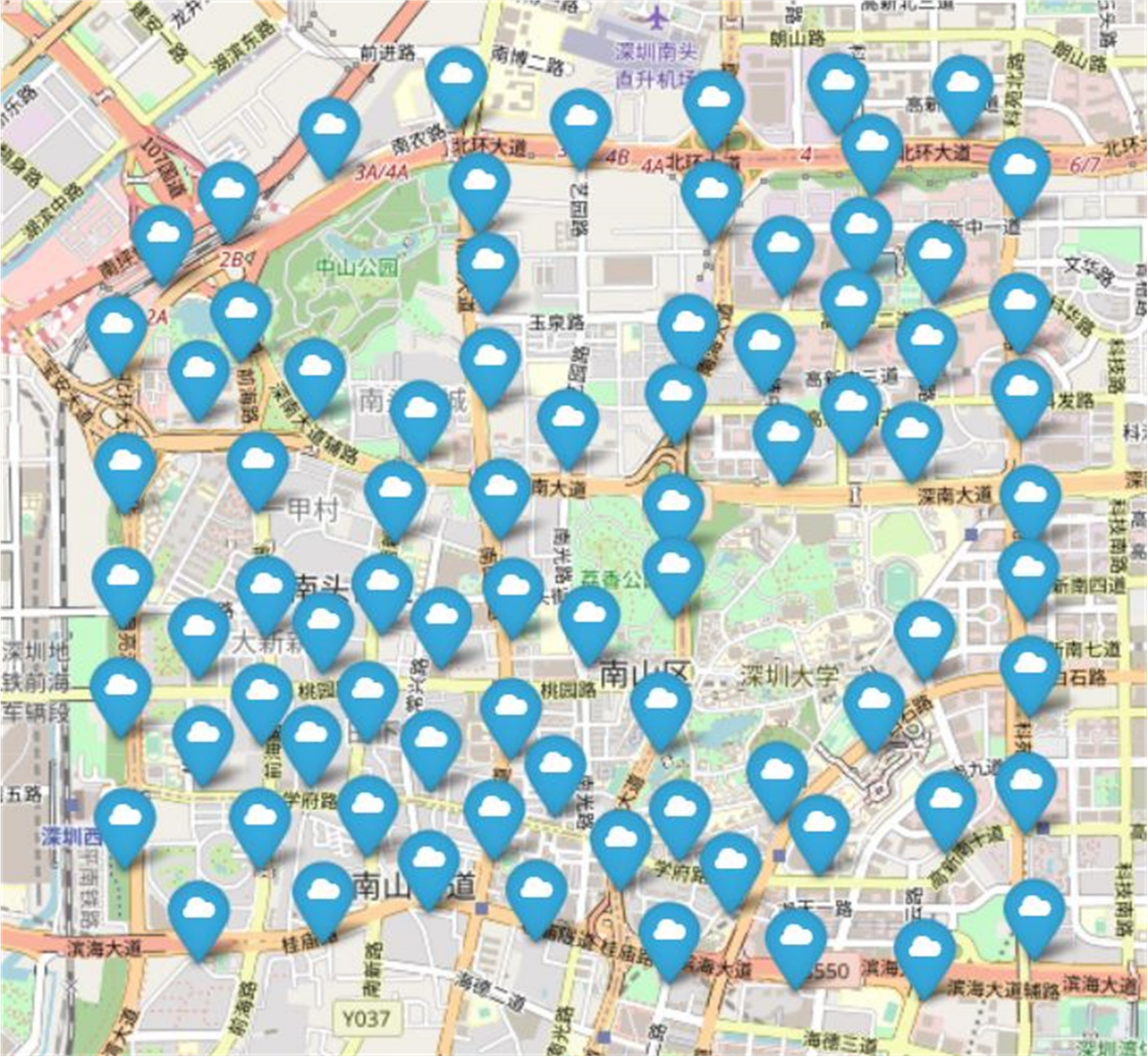}
	\caption{Simulation setup of RSUs placement in Shenzhen.}
	\label{simulation_setup}
\end{minipage}
\begin{minipage}[b]{0.5\textwidth}
	\centering
	\scriptsize
	\begin{tabular}{ |c|c| }
		\hline  
		\textbf{Parameters} & \textbf{Values} \\
		\cline{1-2}
		Number of MBSs, $M$ & 8 \\
		\cline{1-2}
		Number of RSUs, $V$ & 80 \\
		\cline{1-2}
		Number of vehicles, $V$ & 1107 \\
		\cline{1-2}
		Number of available contents, $F$ & 3416\\
		\cline{1-2}
		Size of each content, $s_f$ & 1M\\ 
		\cline{1-2}
		Transmission rate of backhaul link, $\gamma_{CM}$ & 10Mbps \\
		\cline{1-2}
		Transmission rate of Fronthaul link, $\gamma_{MR}$ & 100Mbps \\
		\cline{1-2}
		Transmission rate of the link between MBSs, $\gamma_{MM}$ & 50Mbps \\
		\cline{1-2}
		Radius of coverage range of RSU & 300m \\
		\cline{1-2}
		\hline
	\end{tabular}
	\tabcaption{Simulation Parameters}
	\label{Simulation Parameters}
\end{minipage}
\end{figure}

\subsection{Baseline Schemes and Metrics}
We adopt four baseline schemes for comparison. For baseline 1, we adopt the least recently used (LRU) scheme which is a common caching strategy. In this case, RSUs and MBSs firstly remove the least recently used content in the cache when the limit of cache capacity is reached. For baseline 2, we adopt the random scheme which caches the files randomly. For baseline 3, we evaluate the noncooperative caching performance to indicate the impact of cooperation on the whole system. In this case, all the RSUs and MBSs independently determine its deployment with the prediction of vehicles' future trajectory and content popularity. For baseline 4, it takes cooperative content caching scheme given the prior knowledge of the exact future trajectory and content requests from vehicles, which can be treated as the optimal solution for all the cache schemes.

In the simulation, we mainly compare the performance of different caching strategies in terms of the hit ratio and  average delay.

\subsubsection{Hit Ratio} Hit ratio is the essential metric to evaluate the proposed scheme that measures the effectiveness of a cache decision in fulfilling content requests \cite{Zhengxin Yu}. To show the impact of the cooperative caching, one cache hit means the requested content is delivered by the cache within the cluster, whereas a cache miss means the requested content is not stored in the cluster. Hit ratio is calculated as follows:
\begin{equation}
\textrm{hit ratio} = \frac{\textrm{cache hits}}{\textrm{cache hits + cache miss}}.
\end{equation}

\subsubsection{Average Delay} The latency for each content shown in \eqref{content retrieval delay} is determined by the cache schemes, the retrieval process of HCCN and the unit delays $\gamma_{CM}, \gamma_{MR}, \gamma_{MM}$. Furthermore, as another important metric to evaluate the user experience, the average delay is defined as 
\begin{equation}
\begin{aligned}
&\textrm{average delay} \\
&= \frac{\textrm{Total prefetching delay of all requested contents} }{\textrm{Total number of requested contents} }.    
\end{aligned}
\end{equation}

In the following, we investigate the convergence of cooperative content caching in Subsection \ref{Convergence of Cooperative Content Caching} while the impact of the different system parameters, i.e. MBS size and RSU size, is studied to evaluate hit ratio in Subsection \ref{Hit Ratio Evaluation} and average delay in Subsection \ref{Average Delay Evaluation}.

\subsection{Convergence of Cooperative Content Caching}
\label{Convergence of Cooperative Content Caching}
In Fig. \ref{convergence}, we investigate the impact of different RSU size $S_r^\textrm{RSU}$ on the convergence of the proposed adaptive gradient descent method for large-scale optimization problem. Fig. \ref{convergence} sketches the number of iterations versus the objective function (delay) by considering two cases with configuration given by : i) RSU size = 140, MBS size = 800; 2) RSU size = 410, MBS size = 800. As shown in Fig. \ref{convergence}, the objective presents a tendency to decrease by adjusting the penalty parameter $\beta$ all the time in a practical adaptive gradient descent algorithm. Through subsequent simulations, it can optimize cache deployment efficiently, and thus guarantee a high hit ratio and low delay.
\vspace{-0.8cm} 
\begin{figure}[htbp]
\centering
\includegraphics[scale=0.45]{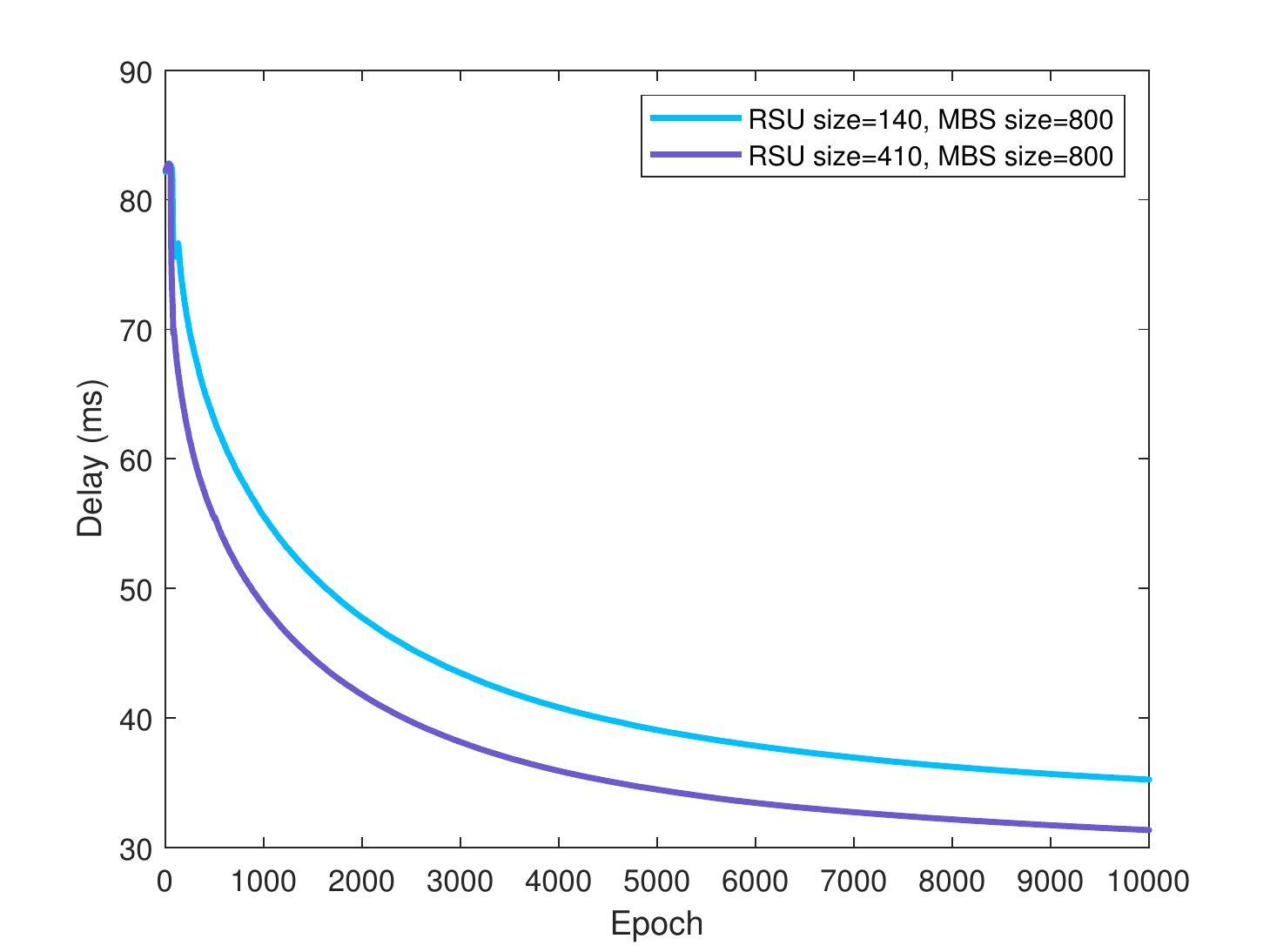}
\caption{Convergence of our adaptive gradient descent method algorithms for different values of $S_r^{RSU}$ when $S_m^{MBS}=800$.}
\label{convergence}
\end{figure}

\subsection{Accuracy of the Trajectory Prediction}
\label{Accuracy of the Prediction Method}
In Fig. \ref{Trajectory Prediction}, we compare the proposed trajectory prediction method with the classical prediction by partial matching (PPM) that conducts prediction of the next location by computing the frequency \cite{ref02_1}. As Fig. \ref{Trajectory Prediction} shows, the accuracy increases by $6\%$ on average by using our proposed method.  The performance gain can be explained as follows. The probable path can be represented by a different RSU sequence because of large coverage area overlap among nearby RSUs. The probability is set as 0 when a given RSU location sequence never occurs in the PPM model, while our method has ability to identify the importance of different location sequences efficiently even a given sequence never occurs. On the other hand, due to the influence of traffic lights, too long residue time may exceed the depth of the tree, which may cause inaccurate predictions with PPM. 
In addition,  it is difficult to obtain the optimal length of paths of the Trie structure in the PPM model in practice. In general, our proposed method can guarantee the effectiveness of trajectory prediction.

\subsection{Effectiveness of the HFL-based SASRec System}
\label{Effectiveness of the HFL-based SASRec System}
In Fig. \ref{HFL-based SASRec System}, we compare the centralized training with HFL to train the SASRec network. We adopt the commonly-used recommendation performance metric, namely, Top-$N$ hit ratio with $N=50$ and 100, which can be denoted as HR@50 and HR@100. For each vehicle, we randomly select 500 negative files, and rank these files with the ground-truth files. As Fig. \ref{HFL-based SASRec System} shows, HR@50 and HR@100 decrease by only $6\%$ and $4\%$, respectively. However, the data privacy can not be guaranteed and the communication overhead is large since the raw data needs to be sent to RSUs and MBSs in the centralized learning manner.

\begin{figure}[htbp]
\centering
\begin{minipage}[t]{0.48\textwidth}
\centering
\includegraphics[scale=0.45]{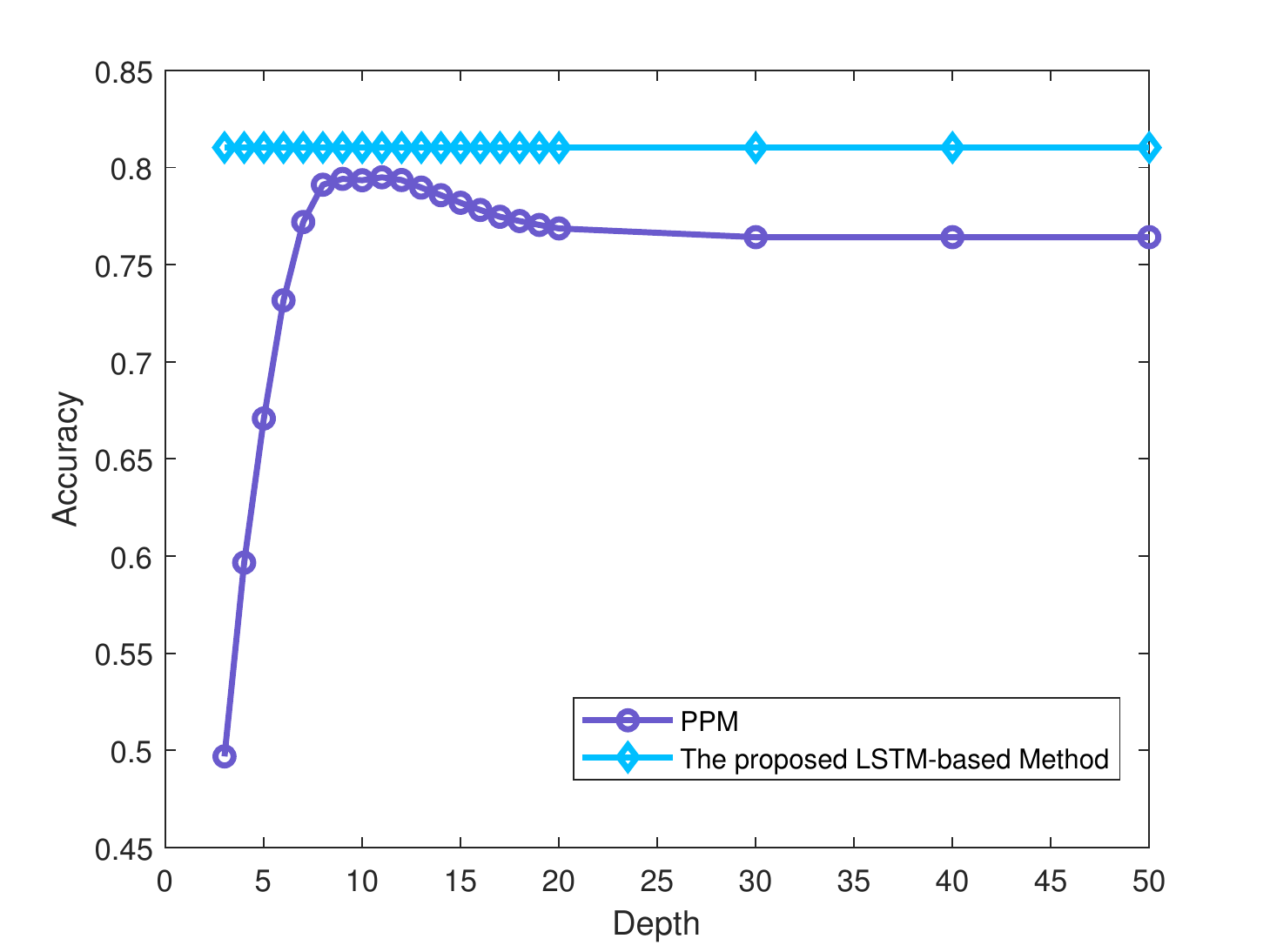}
\caption{Comparison of the proposed trajectory method with PPM.}
\label{Trajectory Prediction}
\end{minipage}
\begin{minipage}[t]{0.48\textwidth}
\centering
\includegraphics[scale=0.45]{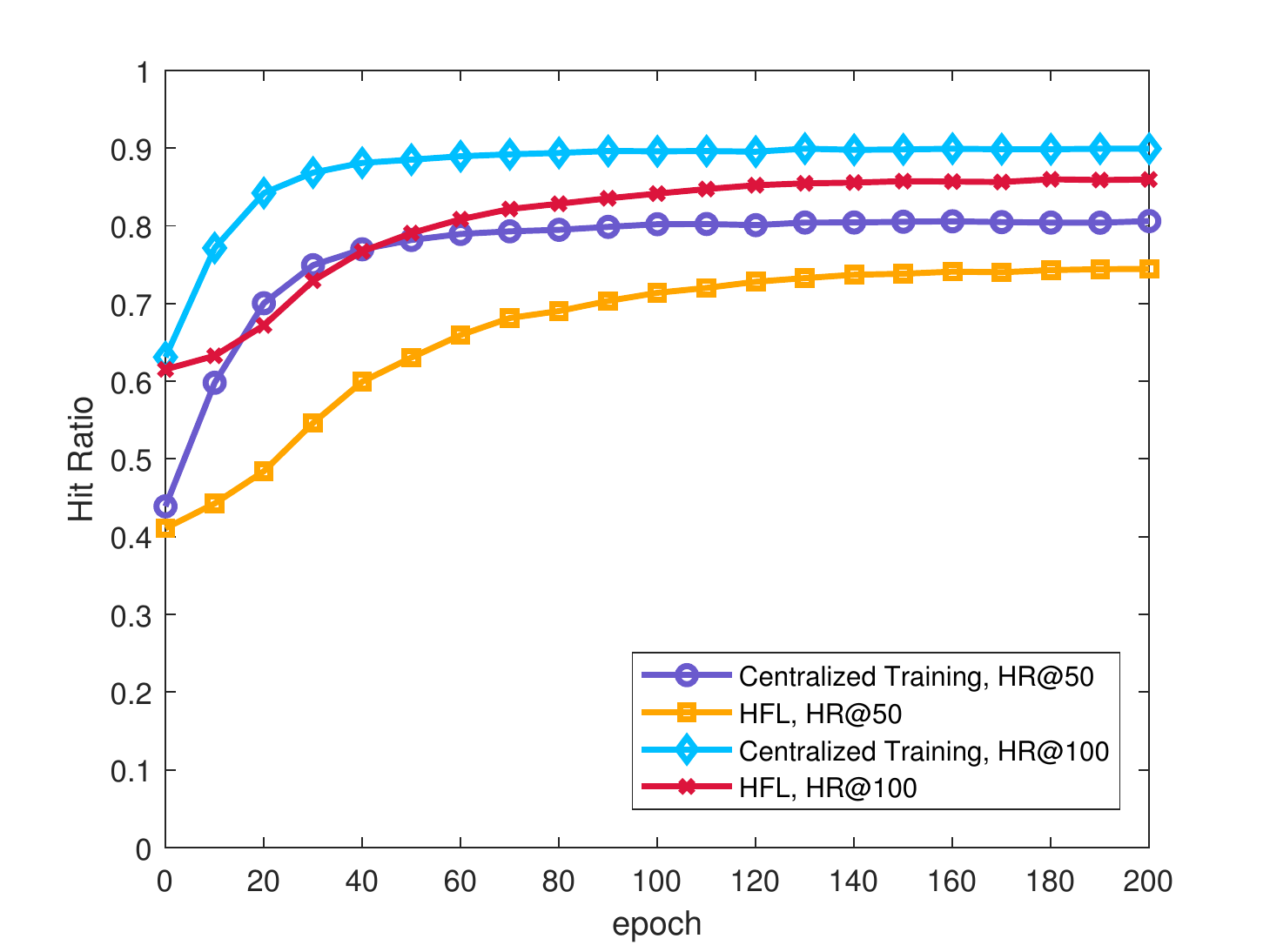}
\caption{Comparison of centralized training with HFL method. We predict the next 5 files for 100 vehicles in the coverage of a BS and 10 RSU are deployed in the edge layer.}
\label{HFL-based SASRec System}
\end{minipage}
\end{figure}

\subsection{Hit Ratio Evaluation}
\label{Hit Ratio Evaluation}
To investigate the impact of caching capacities of RSUs and MBS,  we plot Fig. \ref{Hit_Ratio}(a) to depict the cache hit ratio for varying RSU cache sizes from 230 to 350 contents given 800 contents in MBSs cache, and plot Fig. \ref{Hit_Ratio}(b) to depict the cache hit ratio for varying MBS cache sizes from 300 to 1200 contents with cached 300 contents in RSUs. The results demonstrate that our proposed algorithm outperforms the LRU, random, and non-cooperative caching schemes. With the increase of cache size, the cache hit ratios of all the caching schemes rise. As expected, the lowest cache hit ratio is presented by the classical LRU and random scheme (baseline 1 and baseline 2). The noncooperative caching scheme (baseline 2) outperforms LRU and random scheme because they extract historical trajectory features to predict the future residence time in each RSU, and extract features from the content request history of connected vehicles to predict precise content popularity. Random scheme does not consider any feature of the current environment. LRU only follows static rules without considering dynamically changing content popularity. Since the proposed cooperative caching scheme jointly optimizes the cached contents of all the RSUs, the contents are more likely to be fetched from a neighbor edge node instead of the Internet when it is not cached by the local RSU. Therefore, it can significantly improve resource utilization and show a better performance than the noncooperative caching scheme. Baseline 3 provides the best cache hit ratio since it has the prior knowledge of content requests and trajectory from vehicles in the future, and leverages the advantages of cooperative caching. Compared with LRU caching scheme, the hit ratios are increased $18.8\%$ and $14.7\%$ by using our proposed method in the cases of RSU size = 220, MBS size = 800 and RSU size = 300, MBS sizes = 550, respectively.

\begin{figure*}
\centering
\includegraphics[scale=0.5]{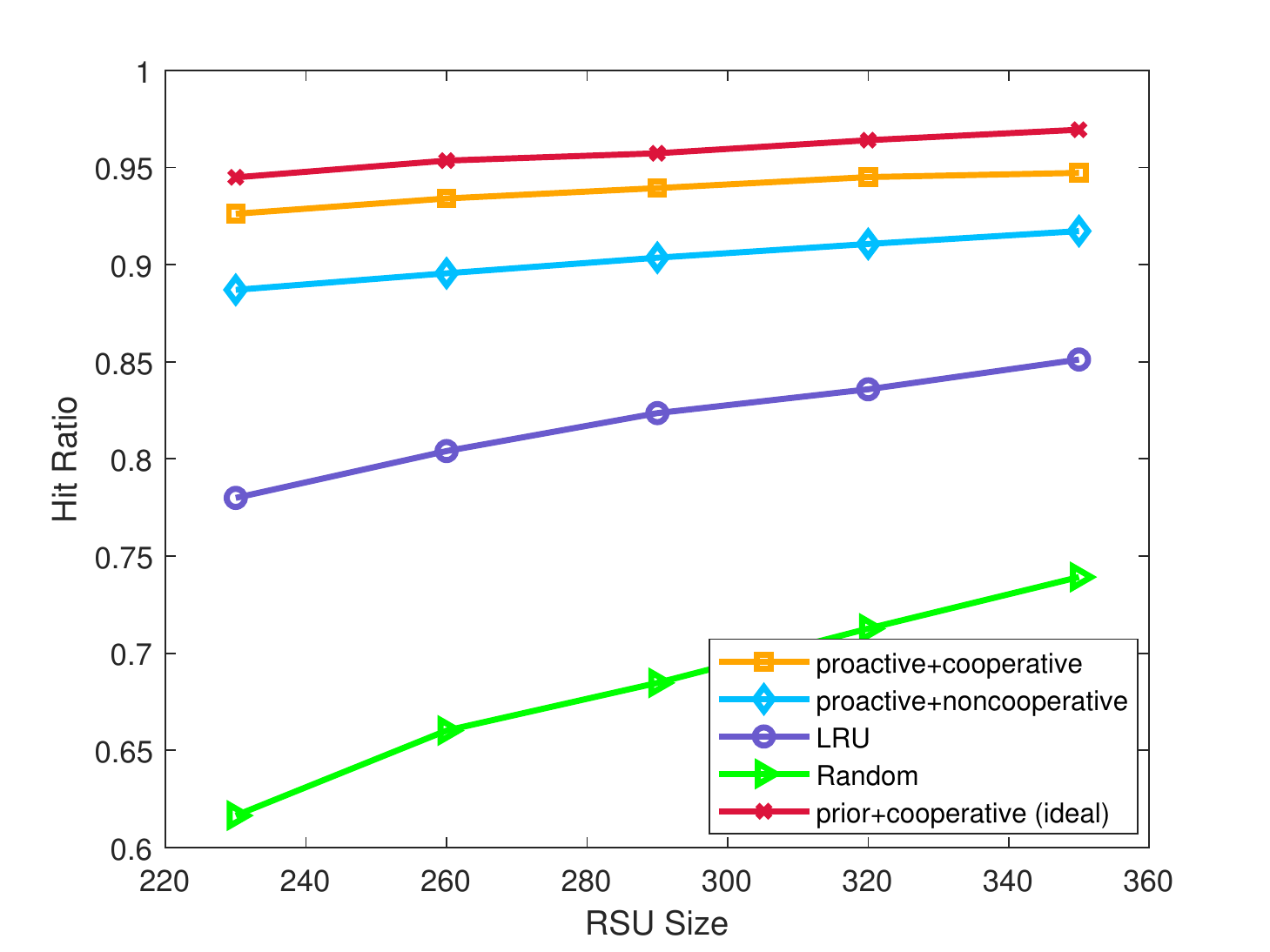}
\hfil
\includegraphics[scale=0.5]{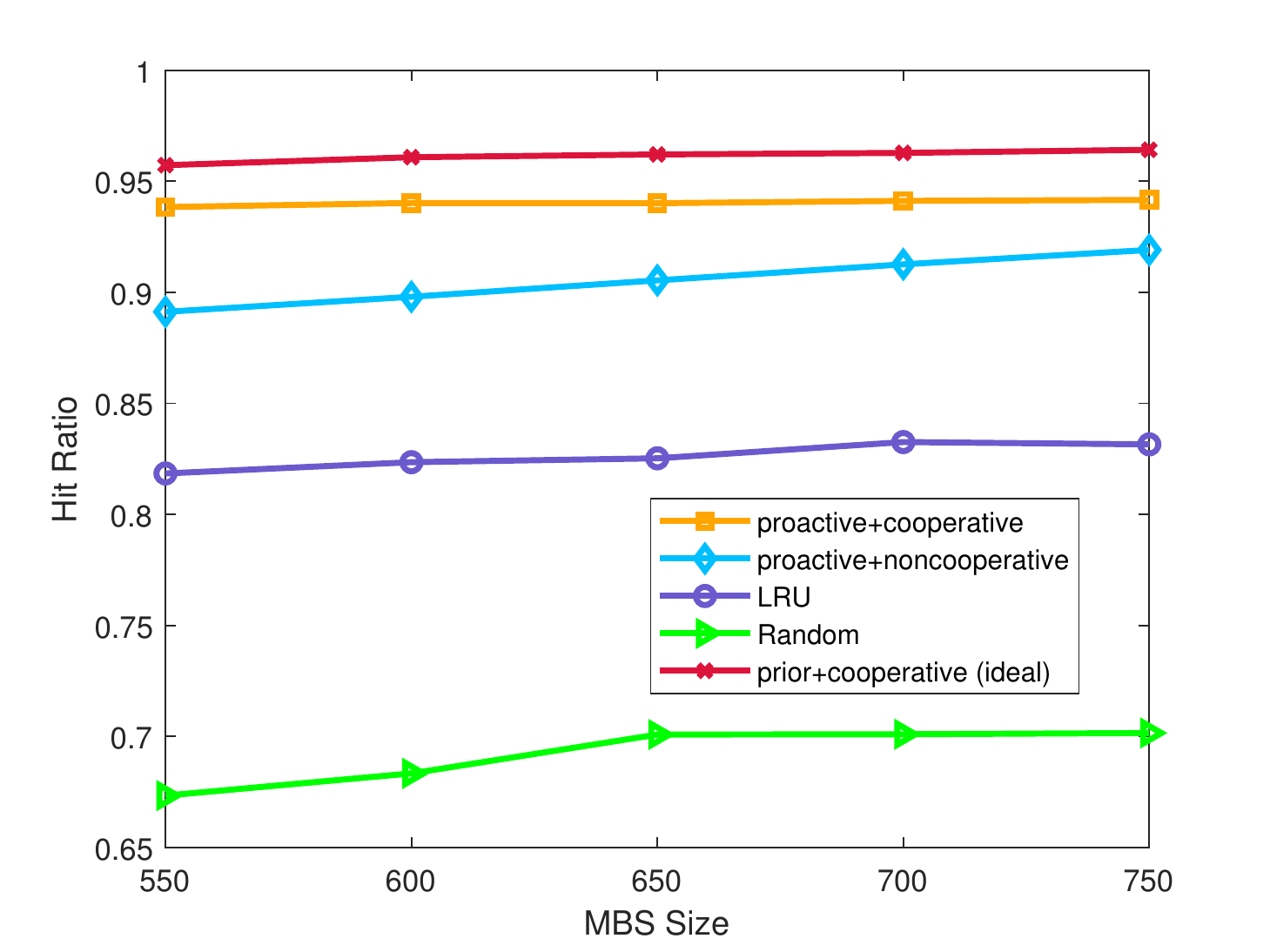}
\caption{(a) Hit ratio versus RSU size in the range 230--350 when MBS size is 800. (b) Hit ratio versus MBS size in the range 550--750 when RSU size is 300.}
\label{Hit_Ratio}
\end{figure*}

\subsection{Average Delay Evaluation}
\label{Average Delay Evaluation}
\begin{figure*}[!t]
\centering
\includegraphics[scale=0.5]{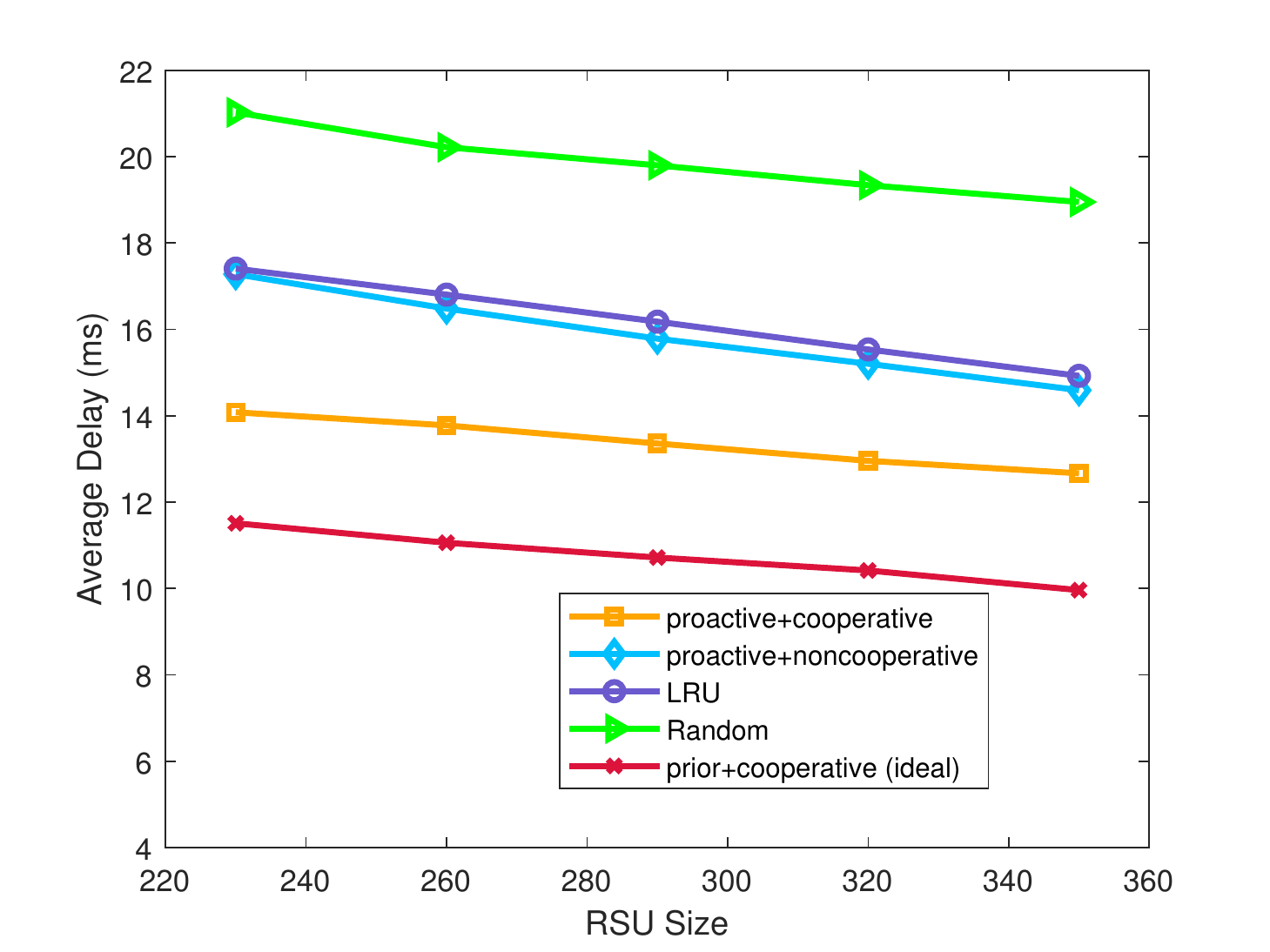} \label{Delay_average_RSU}
\hfil
\includegraphics[scale=0.5]{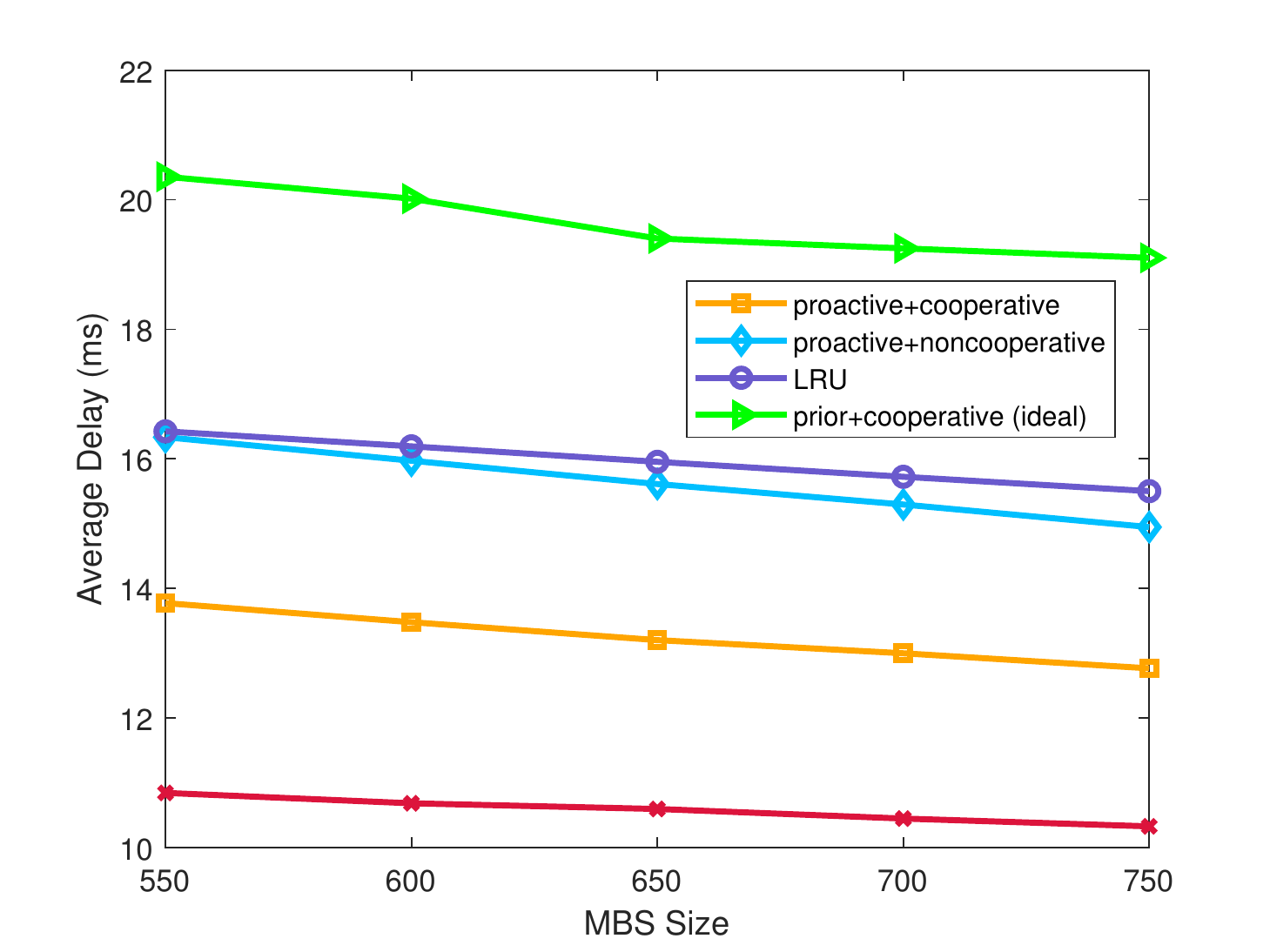} \label{Delay_average_MBS}
\caption{(a) Average delay versus RSU size in the range 230--350 when MBS size is 800. (b) Average delay versus MBS size in the range 550--750 when RSU size is 300.}
\label{Delay_average}
\end{figure*}
To further evaluate the performance of cooperative caching scheme, we plot Fig. \ref{Delay_average}(a) to depict the average delay for varying RSU cache sizes from 230 to 350 contents given 800 contents in MBSs cache, and plot  Fig. \ref{Delay_average}(b) to depict the average delay for varying MBS cache sizes from 300 to 1200 contents with cached 300 contents in RSUs. The results also demonstrate that our proposed algorithm outperforms the LRU, random, and non-cooperative caching schemes in terms of average delay. With the increase of cache size, the average delay of all the caching schemes decline. As expected, the cooperative cache with prior information performs the best, our proposed scheme is the next best, and the classical LRU and random scheme perform worse than any other schemes. The performance gain can be explained by the superiority of cooperative cache scheme. Although the noncooperative cache scheme based on prediction (baseline 2) is slightly better than LRU in terms of average latency, our proposed cooperative cache scheme gains a huger advantage over the LRU. Compared with LRU caching scheme, the average latency is reduced $19.1\%$ and $16.1\%$ by using our proposed method in the cases of RSU size = 220, MBS size = 800 and RSU size = 300, MBS sizes = 550, respectively.

\section{Conclusions}
\label{Conclusion}
In this paper, we propose an HCCN architecture to adapt to the dynamic properties of VANET topology, provide real-time content popularity prediction, and reduce communication costs. And a pipeline scheduling mechanism is utilized to parallelly execute prediction and transmission tasks. To verify the effectiveness of the proposed framework, we simulate the urban roads around Shenzhen University. We firstly make the utmost of the spatio-temporal correlation of historical trajectory data, and then design an LSTM-based model to predict the residence time in each RSU for vehicles in the near future. With the growing concern on data privacy, we propose an HFL-based structure to train the SASRec network for each cluster so as to predict future content popularity in each RSU. Finally, based on the aforementioned trajectory prediction and content popularity prediction results, we propose an adaptive gradient descent-based algorithm to solve a large-scale 0-1 constrained problem and enhance the performance of content caching. Numerical results demonstrate that our proposed cooperative caching scheme achieves a satisfactory performance close to ideal cooperative caching schemes with prior information. Furthermore, we confirm the huge potential of our proposed hierarchical cooperative caching network architecture and a pipeline scheduling mechanism in hit ratio and low latency in future stream media content caching systems.

\appendices
\section{Proof of Proposition \ref{prop0}}
\label{Appendix_prop0}
Due to the multivariate polynomial form of function $\gamma_{r,f}$ in \eqref{total content retrieval delay}, their first and second order partial derivatives are bounded for the variables in the region $[0, 1]$. Based on these bounded partial derivatives, the Hessian of $\gamma_{r,f}$ is a bounded matrix, and the the largest eigenvalue is bounded. Therefore, $\gamma_{r,f}$ has a local Lipschitz continuous gradient w.r.t $\left(\mathbf x, \mathbf y\right)$. Since Sigmoid function also has Lipchitz continuous gradient, $\gamma_{r,f}$ has local Lipschitz continuous gradient w.r.t $\left(\tilde{\mathbf x}, \tilde{\mathbf y}\right)$. Furthermore, the objective function, as a linear combination of $\gamma_{r,f}$, also has local Lipschitz continuous gradient w.r.t $\left(\tilde{\mathbf x}, \tilde{\mathbf y}\right)$. Since $P_r$ and $Q_m$ are linear functions w.r.t $\mathbf x$ and $\mathbf y$, they have local Lipschitz continuous gradients w.r.t $\mathbf x, \mathbf y$. Furthermore, following the similar steps, they have local Lipschitz continuous gradients w.r.t $\tilde{\mathbf x}, \tilde{\mathbf y}$.

\section{Proof of Proposition \ref{prop1}}
\label{Appendix_prop1}
From the fact that the gradient of $W\left(\mathbf{\tilde{x}},\mathbf{\tilde{y}}\right)$ is Lipschitz continuous, we have
\begin{align}
&W\left(\mathbf{\tilde{x}}-\eta\left(\nabla L\right)_{\mathbf{\tilde{x}}},\mathbf{\tilde{y}} -\eta\left(\nabla L\right)_{\mathbf{\tilde{y}}} \right) - W\left(\mathbf{\tilde{x}},\mathbf{\tilde{y}}\right) \notag\\
\leq & -\eta\left[\left(\nabla L\right)_{\mathbf{\tilde{x}}_r}^T\mathbf w_{\mathbf{\tilde{x}}_r} + \left(\nabla L\right)_{\mathbf{\tilde{y}}_m}^T\mathbf w_{\mathbf{\tilde{y}}_m}\right] \notag\\
&\quad\quad\quad+\frac{\lambda_w}{2}\eta^2\bigg[\|\left(\nabla L\right)_{\mathbf{\tilde{x}}_r}\|^2 + \|\left(\nabla L\right)_{\mathbf{\tilde{y}}_m}\|^2\bigg]\notag\\
= & -\eta\bigg[\|\mathbf w\|^2+\beta\bigg(\sum_{r\in\mathcal R} \textrm{ReLU}\left[P_r\left(\tilde{\mathbf x}_r\right)\right] \mathbf p_r^T \mathbf w_{\mathbf{\tilde{x}}_r} \notag\\
&\quad\quad\quad+ \sum_{m\in\mathcal M} \textrm{ReLU}\left[Q_{m}\left(\mathbf{\tilde{y}}_m\right)\right] \mathbf q_{m}^T \mathbf w_{\mathbf{\tilde{y}}_m}\bigg)\\
&\quad\quad\quad-\frac{\lambda_w}{2}\eta\left(\|\left(\nabla L\right)_{\mathbf{\tilde{x}}_r}\|^2 + \|\left(\nabla L\right)_{\mathbf{\tilde{y}}_m}\|^2\right)\bigg]\notag\\
= &-\eta\left[\|\mathbf w\|^2+\beta\phi-\frac{\lambda_w}{2}\eta\left(\|\left(\nabla L\right)_{\mathbf{\tilde{x}}_r}\|^2 + \|\left(\nabla L\right)_{\mathbf{\tilde{y}}_m}\|^2\right)\right]\notag\\
\leq& 0,\notag
\end{align}
where the last inequality holds because of \eqref{obj_eta} and \eqref{obj_beta}, and $W\left(\mathbf{\tilde{x}},\mathbf{\tilde{y}}\right)$ in non- increasing.

\section{Proof of Proposition \ref{prop2}}
\label{Appendix_prop2}
From the fact that the gradient of $P_r$ is Lipschitz continuous, we have
\begin{equation}
\begin{aligned}
&P_r \left(\tilde{\mathbf x}_r-\eta\left(\nabla L\right)_{\mathbf{\tilde{x}}_r}\right) - P_r\left(\tilde{\mathbf x}_r\right) \\
&\leq-\eta\left(\nabla L\right)_{\mathbf{\tilde{x}}_r}^T\mathbf p_r +\frac{\lambda_r}{2}\eta^2\|\left(\nabla L\right)_{\mathbf{\tilde{x}}_r}\|^2\\
& = -\eta(\mathbf w_{\mathbf{\tilde{x}}_r}^T\mathbf p_r + \beta \textrm{ReLU}\left[P_r\left(\tilde{\mathbf x}_r\right)\right] \|\mathbf p_r\|^2-\frac{\lambda_r}{2}\eta\|\left(\nabla L\right)_{\mathbf{\tilde{x}}_r}\|^2) \\
&\leq 0,
\end{aligned}
\end{equation}
where the last inequality holds because of  \eqref{RSU_eta} and \eqref{RSU_beta}, and the objective $W\left(\mathbf{\tilde{x}},\mathbf{\tilde{y}}\right)$ does not increase.
On the other hand, the sufficient condition for the constraint holds is given by
\begin{equation}
\begin{aligned}
P_r \left(\tilde{\mathbf x}_r-\eta\mathbf w_{\mathbf{\tilde{x}}_r}\right)\leq P_r\left(\tilde{\mathbf x}_r\right)-\eta\mathbf w_{\mathbf{\tilde{x}}_r}^T\mathbf p_r + \frac{\lambda_r}{2}\eta^2\|\mathbf w_{\mathbf{\tilde{x}}_r}\|^2\leq 0.
\end{aligned}
\end{equation}
Therefore, the desired condition of $\eta$ in \eqref{RSU_eta_beta} can be derived.

\section{Proof of Proposition \ref{prop3}}
\label{Appendix_prop3}
From the first order Taylor series expansion around $\tilde{\mathbf x}_r$ in RSU caching constraint, we have
\begin{equation}
\begin{aligned}
0 &\geq P_r \left(\tilde{\mathbf x}_r-\eta\left(\nabla L\right)_{\mathbf{\tilde{x}}_r}\right)\\
&= P_r\left(\tilde{\mathbf x}_r\right)-\eta\left(\nabla L\right)_{\mathbf{\tilde{x}}_r}^T\mathbf p_r +o\left(\eta\|\left(\nabla L\right)_{\mathbf{\tilde{x}}_r}\|\right)\\
&=P_r\left(\tilde{\mathbf x}_r\right) - \eta\mathbf w_{\mathbf{\tilde{x}}_r}^T\mathbf p_r - \beta\eta P_r\left(\tilde{\mathbf x}_r\right) \|\mathbf p_r\|^2 \\
&\qquad\qquad\qquad+ o\left(\eta\|\left(\nabla L\right)_{\mathbf{\tilde{x}}_r}\|\right),
\end{aligned}
\end{equation}
which implies that so for positive but sufficiently small $\eta$, 
\begin{equation}
\begin{aligned}
\beta \geq \frac{1}{\eta\|\mathbf p_r\|^2}-\frac{\mathbf w_{\mathbf{\tilde{x}}_r}^T\mathbf p_r}{P_r\left(\tilde{\mathbf x}_r\right) \|\mathbf p_r\|^2}.
\end{aligned}
\end{equation}
Similarly, in MBS caching constraint, $\beta$ should be set
\begin{equation}
\begin{aligned}
\beta \geq \frac{1}{\eta\|\mathbf q_m\|^2}-\frac{\mathbf w_{\mathbf{\tilde{y}}_m}^T\mathbf q_m}{Q_m\left(\tilde{\mathbf y}_m\right) \|\mathbf q_m\|^2}.
\end{aligned}
\end{equation}
Therefore, the desired condition of $\epsilon_r$ and $\epsilon_m$ in \eqref{epsilon_r} can be derived.


\begin{thebibliography}{99}  
\bibitem{ref01} Y. Xing, C. Lv, and D. Cao, ``Personalized vehicle trajectory prediction based on joint time-series modeling for connected vehicles,'' \emph{IEEE Trans. Veh. Technol.}, vol. 69, no. 2, pp. 1341--1352, Feb. 2020.

\bibitem{ref01_1} S. Gambs, M. O. Killijian, and M. N. del Prado Cortez, ``Next place prediction using mobility Markov chains,'' in \emph{Proc. Eur. Conf. Comput. Syst.}, Bern, Switzerland, 2012, p. 3.

\bibitem{ref01_2} W. Mathew, R. Raposo, and B. Martins, ``Predicting future locations with hidden Markov models,'' in \emph{Proc. Ubiquitous Comput.}, Pittsburgh, PA, USA, 2012, pp. 911--918.

\bibitem{ref02} F. Li, Q. Li, Z. Li, Z. Huang, X. Chang, and J. Xia, ``A personal location prediction method based on individual trajectory and group trajectory,'' \emph{IEEE Access}, vol. 7, pp. 92850--92860, Jul. 2019. 

\bibitem{ref02_1} L. Yao, A. Chen, J. Deng, J. Wang, and G. Wu, ``A cooperative caching scheme based on mobility prediction in vehicular content centric networks,'' \emph{IEEE Trans. Veh. Technol.}, vol. 67, no. 6, pp. 5435--5444, June 2018.

\bibitem{ref03} G. Chen, W. Jing, X. Wen, Z. Lu, and S. Zhao, ``An edge caching strategy based on separated learning of user preference and content popularity,'' in \emph{Proc. IEEE/CIC Int. Conf. Commun. China (ICCC)}, Xiamen, China, Nov. 2021, pp. 1018--1023. 

\bibitem{ref04} Y. Jiang, Y. Wu, F. -C. Zheng, M. Bennis, and X. You, ``Federated learning based content popularity prediction in fog radio access networks,'' \emph{IEEE Trans. Wireless Commun.}, vol. 21, no. 6, pp. 3836--3849, Nov. 2021.

\bibitem{ref05} L. Li, C. F. Kwong, Q. Liu, P. Kar, and S. P. Ardakani, ``A novel cooperative cache policy for wireless networks,'' in \emph{Wireless Commun. Mobile Computing}, vol. 2021, pp. 1--18, Aug. 2021. 

\bibitem{ref05_1} A. Chattopadhyay, B. Blaszczyszyn, and H. P. Keeler, ``Gibbsian on-line distributed content caching strategy for cellular networks,'' \emph{IEEE Trans. Wireless Commun.}, vol. 17, no. 2, pp. 969–981, Feb. 2018.

\bibitem{ref06} Z. Zhang, C. Lung, M. St-Hilaire, and I. Lambadaris, ``Smart proactive caching: empower the video delivery for autonomous vehicles in ICN-based networks,'' \emph{IEEE Trans. Veh. Technol.}, vol. 69, no. 7,  pp. 7955--7965, Jul. 2020. 

\bibitem{ref07} L. Hou, L. Lei, K. Zheng, and X. Wang, ``A Q-learning-based proactive caching strategy for non-safety related services in vehicular networks,'' \emph{IEEE Internet Things J.}, vol. 6, no. 3, pp. 4512--4520, Jun. 2019. 

\bibitem{ref07_1} X. Lin, Y. Tang, X. Lei, J. Xia, Q. Zhou, H. Wu, and L. Fan, ``MARL-based distributed cache placement for wireless networks,'' \emph{IEEE Access}, vol. 7, pp. 62606 - 62615, May 2019. 

\bibitem{ref08} R. Kim, H. Lim, and B. Krishnamachari, ``Prefetching-based data dissemination in vehicular cloud systems,'' \emph{IEEE Trans. Veh. Technol.}, vol. 65, no. 1, pp. 292--306, Jan. 2016. 

\bibitem{ref09} Z. Su, Y. Hui, T. H. Luan, and S. Guo, ``Engineering a game theoretic access for urban vehicular networks,'' \emph{IEEE Trans. Veh. Technol.}, vol. 66, no. 6, pp. 4602--4615, Jun. 2017. 

\bibitem{ref10} Y. Wang, Y. Liu, J. Zhang, H. Ye, and Z. Tan, ``Cooperative store-carry-forward scheme for intermittently connected vehicular networks,'' \emph{IEEE Trans. Veh. Technol.}, vol. 66, no. 1, pp. 777--784, Jan. 2017.

\bibitem{RSUupdate1} Y. Guan, X. Zhang, and Z. Guo, ``PrefCache: Edge cache admission with user preference learning for video content distribution,'' \emph{IEEE Trans. Circuits Syst. Video Technol.}, vol. 31, no. 4, pp. 1618--1631, Apr. 2021.

\bibitem{RSUupdate2} D. Huang, X. Tao, C. Jiang, S. Cui, and J. Lu, ``Trace-driven QoE-aware proactive caching for mobile video streaming in metropolis,'' \emph{IEEE Trans. Wireless Commun.}, vol. 19, no. 1, pp. 62--76, Jan. 2020.

\bibitem{RSUupdate3} Z. Hu, Z. Zheng, T. Wang, L. Song, and X. Li, ``Roadside unit caching: Auction-based storage allocation for multiple content providers,'' \emph{IEEE Trans. Wireless Commun.}, vol. 16, no. 10, pp. 6321--6334, Oct. 2017.

\bibitem{Huaqing Wu} H. Wu, J. Chen, W. Xu, N. Cheng, W. Shi, L. Wang, and X. Shen, ``Delay-minimized edge caching in heterogeneous vehicular networks: A matching-based approach,'' \emph{IEEE Trans. Wireless Commun.}, vol. 19, no. 10, pp. 6409--6424, Oct. 2020. 

\bibitem{Zhengxin Yu} Z. Yu, J. Hu, G. Min, Z. Zhao, W. Miao, and M. S. Hossain, ``Mobility-aware proactive edge caching for connected vehicles using federated learning,'' \emph{IEEE Trans. Intell. Transp. Syst.}, vol. 22, no. 8, pp. 5341--5351, Aug. 2021.

\bibitem{DL} I. Goodfellow, Y. Bengio, and A. Courville, \emph{Deep Learning.} Cambridge, MA, USA: MIT Press, 2016.

\bibitem{SASRec} W.-C. Kang and J. McAuley, ``Self-attentive sequential recommendation,'' in \emph{Proc. IEEE Int. Conf. Data Mining (ICDM)}, pp. 197--206, 2018.

\bibitem{HFL} L. Liu, J. Zhang, S. H. Song, and K. B. Letaief, ``Client-edge-cloud hierarchical federated learning,'' in \emph{Proc. IEEE Int. Conf. Commun. (ICC)}, Dublin, Ireland, Jul. 2020, pp. 1--6. 

\bibitem{0-1 KP} D. Li, J. Liu, D. Lee, A. Seyedmazloom G. Kaushik, K. Lee, and N. Park, ``A novel method to solve neural knapsack problem,'' in \emph{Proc. IEEE Int. Conf. Machine Learning (ICML)}, pp. 6414--6424, 2021.

\bibitem{NumOptim} J. Nocedal and S. J. Wright, \emph{Numberical Optimization.} ser. Springer Series in Operations Research. New York: Springer-Verlag, 1999.

\bibitem{SUMO} D. Krajzewicz, ``Traffic simulation with SUMO - Simulation of urban mobility'', \emph{Fundamentals of Traffic Simulation}, 2010.

\bibitem{Movielens} F. M. Harper and J. A. Konstan, ``The movielens datasets: History and context,'' \emph{ACM Trans. Interact. Intell. Syst.}, vol. 5, no. 4, pp. 1--19, Jan. 2016.

\bibitem{NCF} X. He, L. Liao, H. Zhang, L. Nie, and X. Hu, ``Neural collaborative filtering,'' in \emph{Proc. Conf. Int. World Wide Web (WWW)}, 2017, pp. 173--182.

\bibitem{newref1} L. Li, Y. Xu, J. Yin, W. Liang, X. Li, W. Chen, and Z. Han, ``Deep reinforcement learning approaches for content caching in cache-enabled D2D networks,'' \emph{IEEE Internet Things J.}, vol. 7, no. 1, pp. 544--557, Jan. 2020.

\bibitem{newref2} Y. Guo, F. R. Yu, J. An, K. Yang, C. Yu, and V. C. M. Leung, ``Adaptive bitrate streaming in wireless networks with transcoding at network edge using deep reinforcement learning,'' \emph{IEEE Trans. Veh. Technol.}, vol. 69, no. 4, pp. 3879--3892, Apr. 2020.


\end{thebibliography}
\end{document}